\documentclass{article}
\pdfoutput=1
\usepackage{placeins}
\usepackage{authblk}

\usepackage[utf8]{inputenc}
\usepackage{xcolor}
\usepackage{amsmath}
\usepackage{amsthm}
\usepackage{amssymb}
\usepackage[round]{natbib}
\usepackage{microtype}
\usepackage{subfig}
\usepackage{mathtools}
\usepackage[vlined,ruled,linesnumbered]{algorithm2e}
\usepackage{graphics}
\usepackage{pgfplots}
\pgfplotsset{compat=1.4}
\pgfplotsset{filter discard warning=false}
\usepgfplotslibrary{units}
\usepackage{enumerate}
\usepackage{paralist}
\usepackage{booktabs}
\usepackage{tikz}
\usepackage{xspace}
\usetikzlibrary{calc,arrows}
\usepackage[colorlinks, linkcolor=red!40!black, citecolor=red!40!black, bookmarks=false]{hyperref}
\newcommand{\pt}{po\-ly\-nom\-i\-al-time}

\newcounter{mycounter}  
\newenvironment{noindlist}
 {\begin{list}{\arabic{mycounter}.~~}{\usecounter{mycounter} \labelsep=0em \labelwidth=0em \leftmargin=0em \itemindent=\parindent}}
 {\end{list}}

\tikzstyle{vertex}=[circle,draw,fill=black,minimum size=3pt,inner sep=0pt]
\tikzstyle{interval}=[{[-]}]
\tikzstyle{dotinterval}=[circle,draw,minimum size=3pt,inner sep=0pt]

\newcommand{\fp}{fi\-xed-pa\-ra\-me\-ter}
\newcommand{\maxcolcli}{{\Gamma}}

\newcommand{\decprob}[3]{%
\par\vspace{\topsep}%
\begin{minipage}{0.93\linewidth}%
  #1%
  \begin{compactdesc}%
  \item[{Input:}] #2%
  \item[{Question:}] #3%
  \end{compactdesc}%
\end{minipage}\par\vspace{\topsep}%
}

\newcommand{\decprobNN}[3]{%
\par\vspace{\topsep}%
\begin{minipage}{0.93\linewidth}%
  \begin{compactdesc}%
  \item[{Input:}] #2%
  \item[{Question:}] #3%
  \end{compactdesc}%
\end{minipage}\par\vspace{\topsep}%
}

\newcommand{\NP}{\ensuremath{\mathrm{NP}}}
\newcommand{\Wone}{\ensuremath{\mathrm{W}[1]}}

\newcommand{\CompMax}{\ensuremath{c_{\forall}}}
\newcommand{\CompMin}{\ensuremath{c_{\exists}}}

\newcommand{\MColIS}{\textsc{Colorful Independent Set with Lists}\xspace}
\newcommand{\JIntSel}{\textsc{Job Interval Selection}\xspace}

\newcommand{\col}{\ensuremath{\mathrm{col}}}

\newcommand{\cols}{\ensuremath{\gamma}}
\newcommand{\pw}{\ensuremath{\omega}}

\newcommand{\MCIS}{\textsc{2-Union Independent Set}\xspace}

\newcommand{\sig}{\ensuremath{\mathrm{sig}}}

\newtheorem{corollary}{Corollary}
\newtheorem{theorem}{Theorem}
\newtheorem{proposition}{Proposition}
\theoremstyle{definition}
\newtheorem{lemma}{Lemma}
\newtheorem{definition}{Definition}
\newtheorem{rrule}{Reduction Rule}
\newtheorem{observation}{Observation}

\newcommand{\vsig}[1]{\ensuremath{(-#1_s,#1_e,-#1'_s,#1'_e)}}

\newcommand{\poly}{\ensuremath{\operatorname{poly}}}

\newenvironment{comment}{}

\usepackage[textsize=tiny]{todonotes}

\author[1]{René van Bevern} \affil[1]{Institut für Softwaretechnik und Theoretische Informatik, TU~Berlin, Germany, \texttt{\{rene.vanbevern,rolf.niedermeier\}@tu-berlin.de}}

\author{Matthias Mnich}
\affil{Cluster of Excellence Multimodal Computing and Interaction, Saarbr\"ucken, Germany, \texttt{mmnich@mmci.uni-saarland.de}}

\author[1]{Rolf Niedermeier}
\author{Mathias Weller}
\affil{LIRMM, University Montpellier II, France, \texttt{mathias.weller@lirmm.fr}}
\newcommand{\entry}[1]{\ensuremath{T%
\ifthenelse{\equal{#1}{}}{}{[#1]}%
}}

\newcommand{\entryp}[1]{\ensuremath{T'%
\ifthenelse{\equal{#1}{}}{}{[#1]}%
}}

\date{}
\title{Interval Scheduling\\and Colorful Independent Sets\thanks{A
    preliminary version of this article appeared in the proceedings of
    the \emph{23rd International Symposium on Algorithms and Computation
    (ISAAC 2012)}, volume~7676 in Lecture Notes in Computer Science,
    pp. 247--256, Springer, 2012. Besides providing full proof details,
    this revised and extended version improves running times, shows that
    \JIntSel{} is fixed-parameter tractable with respect to the standard
    parameter~$k$, and introduces the parameter
    $c$-compactness. Moreover, it adds an experimental evaluation of the
    algorithms.}}

\begin{document}
\maketitle

{\small\noindent \paragraph{Abstract.}  Numerous applications in scheduling, such as resource allocation or
  steel manufacturing, can be modeled using the NP-hard
  \textsc{Independent Set} problem (given an undirected graph and an
  integer~$k$, find a set of at least~$k$ pairwise non-ad\-ja\-cent
  vertices).  Here, one encounters special graph classes like 2-union
  graphs (edge-wise unions of two interval graphs) and strip graphs
  (edge-wise unions of an interval graph and a cluster graph), on which
  \textsc{Independent Set} remains \NP-hard but admits constant-ratio
  approximations in polynomial time.

We study the parameterized complexity of \textsc{Independent Set} on 2-union graphs and on subclasses like strip graphs. Our investigations significantly benefit from a new structural ``compactness'' parameter of interval graphs and novel problem formulations using vertex-colored interval graphs. Our main contributions are:

\begin{noindlist}
\item We show a complexity dichotomy: restricted to graph classes closed
  under induced subgraphs and disjoint unions, \textsc{Independent Set}
  is po\-ly\-no\-mi\-al-time solvable if both input interval graphs are
  cluster graphs, and is \NP-hard otherwise.

\item We chart the possibilities and limits of effective po\-ly\-no\-mi\-al-time preprocessing (also known as kernelization).

\item We extend \citet{HK06}'s \fp{} algorithm for \textsc{Independent Set} on strip graphs parameterized by the structural parameter ``maximum number of live jobs'' to show that the problem (also known as \JIntSel{}) is \fp{} tractable with respect to the parameter~$k$ and generalize their algorithm from strip graphs to 2-union graphs. Preliminary experiments with random data indicate that \JIntSel{} with up to fifteen jobs and $5\cdot 10^5$ intervals can be solved optimally in less than five minutes.
\end{noindlist}}
%
%
%


%
%
%
%
%
%

\section{Introduction}

Many fundamental scheduling problems can be modeled as finding maximum independent sets in generalizations of interval graphs \citep{KLPS07}.
Intuitively, finding a maximum independent set
corresponds to scheduling a maximum number of jobs (represented by time intervals) on a limited set of machines in a given time frame. %

In this context, we consider two popular generalizations of interval graphs, namely 2-union graphs \citep{BHNSS06} and strip graphs \citep{HK06}:
An undirected graph $G=(V,E)$ is a \emph{2-union graph} if it is the
edge-wise union of two interval graphs~$G_1=(V,E_1)$ and~$G_2=(V,E_2)$
on the same vertex set~$V$, that is, $G=(V,E_1\cup E_2)$, where an
\emph{interval graph} is a graph whose vertices one-to-one correspond to
intervals on the real line and there is an edge between two vertices if
and only if
their intervals intersect. If one of the two interval graphs~$G_1$
or~$G_2$ is even a \emph{cluster graph}, that is, if it consists of
pairwise disjoint cliques, then $G$~is called a \emph{strip graph}.

Examples for solving scheduling problems using (weighted) \textsc{Independent Set} on 2-union graphs include resource allocation scenarios \citep{BHNSS06} and coil
coating in steel manufacturing \citep{HKML11,Moe11}.
Formally, we are interested in the following problem:

\decprob{\MCIS{}}{Two interval graphs~$G_1=(V,E_1),G_2=(V,E_2)$, and a natural number~$k$.}{Is there a size-$k$ independent set in $G=(V,E_1\cup E_2)$?}
\noindent If $G$~is a strip graph, then the problem is known as
\JIntSel{} \citep{Spi99}.  We make two main conceptual contributions:
\begin{noindlist}
\item Since \MCIS{} is NP-hard \citep{BHNSS06}, there is little hope to find optimal solutions within polynomial time. Instead of following the route of approximation algorithms and heuristics \citep{Spi99,BHNSS06%
    ,HKML11}, we aim for solving the problem optimally using \emph{fixed-parameter algorithms} \citep{DF13,FG06,Nie06}, a concept to date largely neglected in the field of scheduling problems \citep{Mar11,MW14}. 

\item In order to obtain our results, we provide ``colorful reformulations'' of \MCIS{} and \JIntSel{}, providing characterizations of these problems in terms of vertex-colored interval graphs, thus replacing the conceptually more complicated 2-union  and strip graphs.
\end{noindlist}

\subsection{Known Results} 

\paragraph{Results for \MCIS{}.}
Checking whether a graph is a 2-union graph is
NP-hard~\citep{GW95,journal-Jia13}.  Therefore, we require two separate
interval graphs as input to \MCIS{}.

To date, a number of polynomial-time approximation algorithms  has been devised to solve \MCIS{}.  \citet{BHNSS06} showed that ver\-tex-weigh\-ted \MCIS{} admits a po\-ly\-no\-mi\-al-time \mbox{ratio-4} ap\-prox\-i\-ma\-tion.  For the special case of so-called \mbox{$K_{1,5}$-free} graphs (which comprises the case that both input graphs are proper interval graphs), \citet{BNR96} provided a ratio-3.25 approximation. 

In the context of applying \MCIS{} to coil coating---a process in steel
manufacturing---\citet{HKML11} showed \NP-hardness of \MCIS{} on
so-called $M$-composite 2-union graphs (which arise in their
application), and showed a dynamic programming based algorithm running
in polynomial time for constant~$M$, where the degree of the polynomial
depends on~$M$. They additionally provided experimental studies based on
heuristics using mathematical programming.

 Regarding parameterized complexity, \citet{Jia10} proved that \MCIS{} is \Wone-hard parameterized by the independent set size~$k$, thus excluding any hope for \fp{} tractability with respect to~$k$. Jiang's \Wone-hardness result holds even when both input graphs are proper interval graphs.

\paragraph{Results for \JIntSel{}.} \JIntSel{} was introduced by \citet{NH82} and was shown $\mathrm{APX}$-hard by \citet{Spi99}, who also provided a ratio-2 greedy approximation algorithm.  \citet{COR06} improved this to a ratio-1.582 approximation algorithm. \citet{HK06} showed \fp{} tractability results for \JIntSel{} in terms of the structural parameter ``maximum number of live jobs'' and in terms of the parameter ``total number of jobs''. Moreover, they showed that recognizing strip graphs is NP-hard.

\subsection{Our Results}

We provide a refined computational complexity analysis for \MCIS{}. Herein, our results mainly touch parameterized complexity. 

We start by proving a complexity dichotomy that shows that all problem variants encountered in our work remain NP-hard: roughly speaking, we show that \textsc{Independent Set} is polynomial-time solvable if the input is the edge-wise union of two cluster graphs, while it is \NP-hard otherwise.

\paragraph{Results for \JIntSel{}.}
  We complement known polynomial-time
approximability results \citep{Spi99,COR06} for \JIntSel{} with parameterized complexity results and extend the tractability
results by \citet{HK06} in several ways:
\begin{noindlist}
\item \looseness=-1 We generalize their \fp{} algorithm for \JIntSel{} parameterized
  by the maximum number of ``live jobs'' to \MCIS{}. Moreover, for
  \JIntSel{}, we show that it can be turned into a \fp{} algorithm with
  respect to the parameter~$k$ (``number of selected intervals'').  Note
  that the latter appears to be impossible for \MCIS{}, which is
  W[1]-hard for the parameter~$k$~\citep{Jia10}.

\item We prove the non-existence of polynomial-size problem kernels for \JIntSel{} with respect to~$k$ and structural parameters like the maximum clique size~$\pw$, thus lowering hopes for provably efficient and effective~preprocessing.

\item We show that, if the input graph is the edge-wise union of a
  cluster graph and a \emph{proper} interval graph, then \JIntSel{}
  admits a problem kernel comprising $4k^2\pw$ intervals that can be
  computed in linear time. %
\end{noindlist}

\paragraph{Results for \MCIS{}.}
Since \MCIS{} is \Wone-hard with respect to the 
parameter~$k$ \citep{Jia10} and NP-hard even when natural graph parameters like ``maximum clique size~$\pw$'' or ``maximum vertex degree~$\Delta$'' are constants (which is implied by our complexity dichotomy), \MCIS{} is unlikely to be \fp{} tractable for any of these parameters.

However, we identify a new natural interval graph parameter that highly
influences the computational complexity of \MCIS{}: we call an interval
graph \emph{$c$-compact} if its intervals are representable using at
most $c$ distinct start and end points.  That is, $c$ is the ``number of
numbers'' required in an interval representation. Similar ``number of
numbers'' parameters have previously been exploited to obtain \fp{}
algorithms for problems unrelated to interval graphs~\citep{FGR12}.

We use $\CompMax$ to denote the minimum number such that \emph{both} input interval graphs are $\CompMax$-compact and $\CompMin$ to denote the minimum number such that at least \emph{one} input interval graph is $\CompMin$-compact. We obtain the following results:

\begin{noindlist}
\item\looseness=-1 We give a simple polynomial-time data reduction rule for \MCIS{}. The analysis of its effectiveness naturally leads to the compactness parameter: the reduction rule yields a $\CompMax^3$-vertex problem kernel. This improves to a $2\CompMax^2$-vertex problem kernel if one of the input graphs is a proper interval graph.
\item The problem kernel with respect to~$\CompMax$ shows that \MCIS{} is \fp{} tractable with respect to~$\CompMax$. By generalizing \citet{HK06}'s \fp{} algorithm from \JIntSel{} to \MCIS{}, we improve this to a time-$O(2^{\CompMin}\cdot n)$ \fp{} algorithm for the  parameter~$\CompMin\leq \CompMax$. %
\end{noindlist}

\begin{table*}
  \centering
  \caption[Overview of parameterized complexity results]{Overview of parameterized complexity results for \JIntSel{}, where~$G_2$---one of the two input graphs---is a cluster graph, and \MCIS, where $G_2$~is any interval graph.  Results for various graph classes of~$G_1$---the other input graph---are shown.    The complexity dichotomy in \autoref{thm:dichotomy} shows that all these problem variants remain NP-hard.}
  \label{tab:results}\small
  \begin{tabular}{lp{4.4cm}p{4.4cm}}
    \toprule
    Class of~$G_1$&\JIntSel{}&\MCIS{}\\
    \midrule
    interval
    &randomized FPT algorithm:
    $O(5.5^k\cdot n)$ time (\autoref{thm:colorcoding})\newline \newline No polynomial-size kernel
    w.\,r.\,t.\ $k$ and~$\omega$ (\autoref{thm:uig-nopoly}) 
    &  FPT algorithm:
    $O(2^{\CompMin}\cdot n)$~time (\autoref{thm:fpt-2union})\newline\newline problem kernel:
    $\CompMax^3$ vertices in $O(n\log^2n)$~time (\autoref{thm:sigkern})\\
    \addlinespace[1em]
    proper
    
    interval 
    & problem kernel:
    $4k^2\omega$ vertices 
    in $O(n)$~time (\autoref{thm:splitUIG kernel})\newline
    & problem kernel:
    $2\CompMax^2$ vertices in $O(n\log^2n)$~time (\autoref{thm:sigkern})\\
    \bottomrule
  \end{tabular}
\end{table*}

\bigskip\noindent \autoref{tab:results} summarizes our results.
Experiments with random data indicate that, within less than
five minutes, one can optimally solve \JIntSel{} with up to fifteen jobs
and $5\cdot 10^5$ intervals and \MCIS{} with $\CompMin\leq 15$ and
$5\cdot 10^5$~intervals.

\paragraph{Organization of this Work.} %
In \autoref{sec:realpreliminaries}, we introduce basic notation and the concepts of parameterized algorithmics.

 \autoref{sec:preliminaries} introduces the compactness parameter for interval graphs and some basic observations on compactness.  In the remaining sections, we assume to work on $c$-compact representations of interval graphs such that $c$~is minimum. 

\looseness=-1\autoref{sec:colors} presents our colored model of \MCIS{} and \JIntSel{} and discusses pros and cons of the new model. 

\looseness=-1 \autoref{sec:dichotomy} presents a computational complexity dichotomy that has consequences both for \JIntSel{} and \MCIS{}.

\autoref{sec:jisp} presents our results specific to \JIntSel{}, whereas  \autoref{sec:mcis} contains the results for the more general \MCIS{}. 

Finally, we present experimental results in \autoref{sec:exp} and conclude in \autoref{sec:discussion}.

\section{Preliminaries}\label{sec:realpreliminaries}

Throughout the work, we use the notation~$[c]$ as shorthand for the subset~$\{1,2,\dots,c\}$ of natural numbers.

We consider undirected, finite graphs~$G=(V,E)$ with vertex set~$V(G)$
and edge set~$E(G)$. If not stated otherwise, we use $n\coloneqq{}|V|$
and~$m\coloneqq{}|E|$. Two vertices~$v,w\in V$ are \emph{adjacent} or
\emph{neighbors} if~$\{v,w\}\in E$.  The \emph{open
  neighborhood}~$N_G(v)$ of a vertex~$v\in V$ is the set of vertices
that are adjacent to~$v$, the \emph{closed neighborhood}
is~$N_G[v]\coloneqq N_G(v)\cup\{v\}$.  For a vertex set~$U\subseteq V$,
we define~$N_G[U]\coloneqq{}\bigcup_{v\in U} N_G[v]$.
If the graph~$G$ is clear from context, we
drop the subscript~$G$.  For a vertex set~$V'\subseteq V$, the
\emph{induced subgraph}~$G[V']$ is the graph obtained from~$G$ by
deleting all vertices in $V\setminus V'$.

An \emph{independent set} is a set of pairwise non-adjacent vertices.
A \emph{matching} is a set of pairwise disjoint edges.
 The \emph{chromatic index~$\chi'(G)$} of~$G$ is the minimum number of colors required in a \emph{proper edge coloring}, that is, in a coloring of edges of~$G$ such that no pair of edges sharing a vertex has the same color.

A \emph{path} in~$G$ from~$v_1$ to~$v_\ell$ is a sequence~$(v_1,v_2,\dots,v_\ell)\in V^{\ell}$ of vertices with~$\{v_i,v_{i+1}\}\in E$ for~$i\in[\ell-1]$.
Its \emph{length} is ${\ell-1}$. We denote a path on~$\ell$ vertices by~$P_{\ell}$. Two vertices~$v$~and~$w$ are \emph{connected} in~$G$ if there is a path from~$v$ to~$w$ in~$G$. A \emph{connected component} of~$G$ is a maximal set of pairwise connected vertices. If in each connected component of~$G$, all its vertices are pairwise adjacent (that is, they form a \emph{clique}), then we call $G$~a \emph{cluster graph}. Equivalently, a graph is a cluster graph if and only if it does not contain a~$P_3$ as induced subgraph.

The \emph{disjoint union} of two graphs~$G_1=(V_1,E_1)$ and $G_2=(V_2,E_2)$ is the graph~$G_1\uplus G_2=(V_1\uplus V_2, E_1\uplus E_2)$, where  $V_1\cap V_2=\emptyset$. The \emph{edge-wise union} of two graphs~$G_1=(V,E_1)$ and~$G_2=(V,E_2)$ on the same vertex set is $G_1\cup G_2=(V,E_1\cup E_2)$. A~class of graphs~$\mathcal C$ is \emph{closed under induced subgraphs} if $G=(V,E)\in\mathcal C$ implies $G[V']\in\mathcal C$ for any $V'\subseteq V$. A~class of graphs~$\mathcal C$ is \emph{closed under disjoint unions} if $G_1,G_2\in\mathcal C$ implies $G_1\uplus G_2\in C$.

\bigskip\noindent An \emph{interval graph} is a graph whose vertices can be represented as closed %
intervals on the real line such that two vertices $v$~and~$w$ are adjacent if and only if the intervals corresponding to $v$~and~$w$ intersect. We denote the start point of the interval associated with~$v$ by~$v_s$ and its end point by~$v_e$.
 A~graph is a \emph{proper interval} graph if it allows for an interval representation such that for no two intervals~$v$ and~$w$ it holds that $v\subsetneq w$.
Equivalently, proper interval graphs are precisely those interval graphs that do not contain a~$K_{1,3}$ as induced subgraph~\citep{BLS99}.

\paragraph{Fixed-Parameter Algorithms.} The main idea in \fp{} algorithms is to accept exponential running time, which is seemingly inevitable in solving NP-hard problems, but to restrict it to one aspect of the problem, the \emph{parameter}. More precisely, a problem $\Pi$ is \emph{\fp{} tractable~(FPT)} with respect to a parameter~$k$ if there is an algorithm solving any instance of $\Pi$ with size~$n$ in $f(k) \cdot \poly(n)$~time for some computable function~$f$ \citep{DF13,FG06,Nie06}.  Such an algorithm is potentially efficient for small values of~$k$.

\paragraph{Problem Kernelization.} One way of deriving \fp{} algorithms is \emph{problem kernelization}  \citep{GN07,Kra14}. As a formal approach of describing efficient data reduction that preserves optimal solutions, problem kernelization is a powerful tool for attacking NP-hard problems.
A \emph{kernelization algorithm} consists of polynomial-time executable \emph{data reduction rules} that, applied to any instance $x$ with parameter~$k$, yield an equivalent instance $x'$ with parameter~$k'$,
such that both the size $|x'|$ and $k'$ are bounded by some functions $g$ and $g'$ in~$k$, respectively. The function~$g$ is referred to as the \emph{size} of the \emph{problem kernel}~$(x',k')$.  Mostly, the focus lies on finding problem kernels of polynomial size.

\section{Compact Interval Graphs}\label{sec:preliminaries}

A parameter that we will see to highly influence the computational complexity of \MCIS{} is the ``compactness'' of an interval graph, which corresponds to the number of distinct numbers required in its  interval representation. %

\begin{definition}\label{def:compactness}
  An interval representation is \emph{$c$-compact} if the start and end point of each interval lies in~$[c]$. Moreover, the intervals are required to be sorted by increasing start points.

  An interval graph is \emph{$c$-compact} if it admits a $c$-compact interval representation.
\end{definition}

\noindent In this work, we state all running times under the assumption that a \emph{$c$-compact interval representation} for minimum~$c$ is given. We show in the following that such a representation can be efficiently computed. To this end, we make a few observations.
\begin{observation}\label{obs:cliques=compactness}
  Let $G$ be an interval graph and $c$~be the minimum integer such that $G$~is $c$-compact. Then $G$ has exactly~$c$~maximal cliques.
\end{observation}

\begin{proof}
  Let $c'$ be the number of maximal cliques in~$G$. We show $c=c'$ by proving $c'\leq c$ and $c\leq c'$ independently.

  First, it is easy to see that a $c$-compact interval graph has at most $c$~maximal cliques: each interval end point~$v_e$ gives rise to at most one maximal clique, which consists of the intervals containing the point~$v_e$. Hence, $c'\leq c$.

  Second, the interval graph~$G$ allows for an ordering of its~$c'$~maximal cliques such that the cliques containing an arbitrary vertex occur consecutively in the ordering \citep{journal-FG65}. Hence, a $c'$-compact interval representation can be constructed in which each vertex~$v$ is represented by the interval $[v_s,v_e]$, where $v_s$ is the number of the first maximal clique containing~$v$ and $v_e$ is the number of the last maximal clique containing~$v$. It follows that $c\leq c'$. 
\end{proof}

\noindent From \autoref{obs:cliques=compactness}, it immediately follows that

\begin{observation}
  An $n$-vertex interval graph is $n$-compact.
\end{observation}

\noindent\looseness=-1 In the remainder of this article, we will assume to be given a $c$-compact representation for \emph{minimum}~$c$. This assumption is justified by the fact that such a $c$-compact representation is computable in $O(n\log n)$~time from an arbitrary interval representation or even in linear time from a graph given as adjacency list.

\begin{observation}\label{obs:shrinkints}
  Any interval representation of an interval graph~$G$ can be converted
  into a $c$-compact representation for~$G$ in $O(n\log n)$~time such that
  \begin{enumerate}[i)]
  \item\label{shrinkints2} at each position in~$[c]$, there is an
    interval start point \emph{and} an interval end point, and
  \item\label{shrinkints1} $c$~is the minimum number such that~$G$ is
    $c$-compact.
  \end{enumerate}
\end{observation}

\begin{proof}
  We first sort all \emph{event points} (start or end points of intervals) in increasing order in $O(n\log n)$~time. Then, in linear time, we iterate over all event points in increasing order and move each event point to the smallest possible integer position that maintains all pairwise intersections.  It remains to show (i) and (ii).

  (\ref{shrinkints2}) First observe that every interval start point~$v_s$ is also an end point for some interval: otherwise, we would have moved the event point~$v'_e$ (possibly, $v'_e=v_e$) that directly follows $v_s$ to the position~$v'_e-1$, maintaining all pairwise intersections. It follows that $G$~is $c'$-compact, where $c'$ is the number of different end point positions.

  Second, every interval end point~$v_e$ is also a start point for some interval: otherwise, we could have moved~$v_e$ to the position~$v_e-1$ maintaining all pairwise intersections. It follows that each end point~$v_e$ gives rise to a distinct maximal clique, because the interval starting at~$v_e$ cannot be part of the maximal cliques raised by earlier end points.

  (\ref{shrinkints1}) From~(\ref{shrinkints2}), it follows that, for
  any two positions~$i,j$ of event points, the set of intervals
  containing~$i$ and the set of intervals containing~$j$ are distinct
  and, therefore, $i$~and~$j$ give rise to distinct maximal cliques
  in~$G$. Thus, our algorithm computes a $c'$-compact representation
  of~$G$ with $c'\leq c$, where $c$ is the number of maximal cliques
  in~$G$. From \autoref{obs:cliques=compactness}, it follows that $c'$
  is the minimum number such that $G$~is $c'$-compact.
\end{proof}

\noindent If the input graph is given in form of an adjacency list, we can compute a $c$-compact representation for minimum~$c$ in linear time.

\begin{observation}
  Given an interval graph~$G$ as adjacency list, a $c$-compact representation for minimum~$c$ can be computed in $O(n+m)$~time.
\end{observation}

\begin{proof}
  Using a linear-time algorithm by \citet[Section~8]{COS09}, we obtain an $n$-compact interval representation of~$G$. Using this, we can execute the algorithm in the proof of \autoref{obs:shrinkints} in linear time, since the list of sorted event points can be obtained in linear time using counting sort: the list has $n$~elements and the sorting keys are integers not exceeding~$n$.
\end{proof}

\section{Colorful Independent Sets}\label{sec:colors}
Many of our results significantly benefit from a novel but
natural embedding of \MCIS{} into a more general problem: \MColIS{}. We discuss this embedding in the following.

\subsection{Colorful Independent Sets and Job Interval Selection}
The first step in formalizing \MCIS{} as \MColIS{} is an
alternative formulation of
the classical scheduling problem \JIntSel{}.

The task in \JIntSel{} is to execute a maximum number of jobs out of a given set, where each job has multiple possible execution intervals, each job is executed at most once, and a machine can only execute one job at a time. We formally state this problem in terms of colored interval graphs, where the colors correspond to jobs and intervals of one color correspond to multiple possible execution times of the same job:

\decprob{\JIntSel{}}{An interval graph~$G=(V,E)$, a coloring $\col\colon V\rightarrow[\cols]$, and a natural number~$k$.}{Is there a size-$k$ colorful independent set in~$G$?}

\looseness=-1\noindent Here, \emph{colorful} means that no two vertices of the independent set
have the same color. 

Note that the colored formulation of \JIntSel{} is indeed equivalent to the known formulation \citep{Spi99} as special case of \MCIS{}, where one input interval graph is a cluster graph:

\decprobNN{\JIntSel{}$\ast$}{An interval graph~$G_1=(V,E_1)$, a cluster graph $G_2=(V,E_2)$, and a natural number~$k$.}{Is there a size-$k$ independent set in $G=(V,E_1\cup E_2)$?}

\noindent In this second formulation, the maximal cliques of $G_2$ correspond to jobs, and the intervals in $G_1$ that are part of the same maximal clique in $G_2$ correspond to multiple possible execution times of the same job. That is, the maximal cliques in $G_2$ one-to-one correspond to the colors in the colorful problem formulation.

\medskip\noindent In \autoref{sec:DP}, we restate the \fp{} algorithms for \JIntSel{} by \citet{HK06} in terms of our colorful formulation. This formulation uncouples the algorithms from the geometric arguments originally used by \citet{HK06} and allows for a more combinatorial point of view. Exploiting this, we turn \citet{HK06}'s \fp{} algorithms for the total number of jobs (which translates to the number~$\gamma$ of colors in our formulation) into a \fp{} algorithm for the smaller parameter~$k$---the number of jobs we want to execute.

\subsection{From Strip Graphs to 2-Union Graphs}\label{sec:mcolisdef}
Our  more combinatorially stated version of \citet{HK06}'s \fp{} algorithm easily applies to \MColIS{}, which is a canonical generalization of \JIntSel{}:

\decprob{\MColIS{}}{An interval graph~$G=(V,E)$, %
      a \emph{list-color\-ing} $\col\colon V\rightarrow2^{[\cols]}$, and a natural number~$k$.}{Is there a size-$k$ colorful independent set in~$G$?}

\noindent Here, \emph{colorful} means that the 
intersection of the color sets of any two vertices in the independent set
is empty. 

We will later show that \MColIS{} is actually even more general than \MCIS{}. The colored reformulation turned out to be the key in generalizing \citet{HK06}'s algorithm for \JIntSel{} to \MCIS{}.

\subsection{Advantages and Limitations of the Model}
The colorful formulation of \JIntSel{} helps us to transform
\citet{HK06}'s \fp{} algorithm for the parameter
``number~$\gamma$ of colors'' into a \fp{} algorithm
for the parameter ``size~$k$ of the sought colorful
independent set'', where $\gamma\leq k$.  Moreover, the algorithm for the colorful formulation
of \JIntSel{} easily generalizes to \MColIS{} and, as we will see, to
\MCIS{}.

The advantage of considering \MColIS{} instead of \MCIS{} is that one
can concentrate on a single given interval graph instead of two merged
ones, thus making the numerous structural results on interval graphs
applicable.  Possibly, the colorful view on finding independent sets and
scheduling might be useful in further studies.

Not always, however, the colorful viewpoint is superior to the geometric one. Herein, it is important to note that \MColIS{} is actually a more general problem than \MCIS{} and it is cumbersome to formulate precisely \MCIS{} in terms of \MColIS{}. Thus, when exploiting the specific combinatorial properties of \MCIS{}, for example in the kernelization algorithm in \autoref{sec:independentsetin2uniographs}, the colored model is not helpful.

Moreover, in the following \autoref{sec:cindexdic}, we prove hardness
results for finding independent sets not only on 2-union and strip
graphs.  Hence, the colorful model is not exploited
there. %

\section{A Complexity Dichotomy}\label{sec:cindexdic}
\label{sec:dichotomy}
In this section, we determine the computational complexity of  \textsc{Independent Set} on edge-wise unions of graphs in dependence of the allowed input graph classes. Formally, we define the considered problem as follows:
\newcommand{\CIS}{\textsc{Common Independent Set}}

\decprob{\CIS{}}{Two graphs~$G_1=(V,E_1),G_2=(V,E_2)$, and a natural number~$k$.}{Is there a size-$k$ independent set in $G=(V,E_1\cup E_2)$?}

\noindent Note that \CIS{} contains \MCIS{} as special case since the
only difference is that it does not restrict the two input graphs to be
interval graphs.

If we assume that the input graphs $G_1$~and~$G_2$ are members in a graph class that is closed under induced subgraphs and disjoint unions (as, for example, the widely studied chordal graphs and, in particular, interval graphs, cluster graphs, and forests \citep{BLS99}), we can use the main result of this section to precisely state for which classes \CIS{} is NP-hard and for which it is polynomial-time solvable, thus giving a \emph{complexity dichotomy} of the problem. 

\bigskip\noindent Now, we first precisely state the dichotomy theorem. Since it is quite technical, we immediately illustrate it by some examples in form of implications to the complexity of \JIntSel{} and \MCIS{}. We conclude the section with the proof of the theorem.
\begin{theorem}
\label{thm:dichotomy}
  Let $\mathcal C_1$ and $\mathcal C_2$ be graph classes such that
  \begin{itemize}
  \item $\mathcal C_1$ and $\mathcal C_2$ are closed under disjoint unions and induced subgraphs, and
  \item $\mathcal C_1$ and $\mathcal C_2$ each contain at least one graph that has an edge.
  \end{itemize}
Then, \CIS{} restricted to input graphs $G_1\in\mathcal C_1$ and $G_2\in\mathcal C_2$
  \begin{enumerate}[i)]
  \item is solvable in $O(n^{1.5})$~time if $\mathcal C_1$ and $\mathcal C_2$ only contain cluster graphs, and
  \item NP-hard otherwise.
  \end{enumerate}
\end{theorem}

\noindent 
We illustrate \autoref{thm:dichotomy} using some examples.
First, choose $\mathcal C_1=\mathcal C_2$ to be the class of interval graphs that have maximum degree~$2$ and maximum clique size~$2$. Since~$P_3$ is an interval graph but not a cluster graph in $\mathcal C_1$, we obtain the following NP-hardness result for \MCIS{}:

\begin{corollary}\label{cor:mcis-veryhard}
  \MCIS{} is NP-hard even if both of the following hold:
  \begin{itemize}
  \item the maximum degree of each input interval graphs is~$2$,

  \item the maximum clique size of each input interval graph is~$2$.
  \end{itemize}
\end{corollary}

\noindent Now, choose $\mathcal C_1$ to be the class of disjoint unions of paths of length at most two and~$\mathcal C_2$ to be the class of cluster graphs of clique size at most two. Then, we have $P_3\in\mathcal C_1$, which is not a cluster graph and, since disjoint unions of paths are interval graphs, we obtain the result that \MCIS{} is NP-hard even if one input interval graph is a cluster graph consisting of cliques of size two and if the other is a disjoint union of paths of length at most two. Transferring this to the colored model as described in \autoref{sec:colors}, we obtain:

\begin{corollary}
  \JIntSel{} is NP-hard even if the input interval graph is a disjoint union of paths of length at most two and contains each color at most two times.
\end{corollary}

\noindent It remains to prove \autoref{thm:dichotomy}. %
For the first part of the theorem, \citet{BHNSS06} and \citet{HK06}
mentioned that \CIS{} is \pt{} solvable if both input graphs are cluster
graphs. We prove here an explicit upper bound on the running time.
\begin{lemma}\label{lem:clusolve}
  \CIS{} is solvable in $O(n^{1.5})$ time if both input interval graphs are cluster graphs.
\end{lemma}
\begin{proof}
  Let~$G_1=(V,E_1),G_2=(V,E_2)$ be cluster graphs on $n$~vertices, and let~$G=(V,E_1\cup E_2)$. %
  Let $C_i$ denote the set of connected components (cliques) of~$G_i$ for $i \in\{1,2\}$.
  Define a bipartite graph~$H=(C_1\uplus C_2, E_H)$ with edge set~$E_H$ such that there is an edge~$\{U_1,U_2\}\in E_H$ if and only if there is a vertex~$v\in V$ that occurs in clique~$U_1$ of~$G_1$ and in clique~$U_2$ of~$G_2$.
  We claim that~$G$ has an independent set of size~$k$ if and only if~$H$ has a matching of size~$k$.

  First, let $I$~be an independent set of size~$k$ for~$G$.
  We construct a matching~$M$ in~$H$.
  To this end, for each~$v\in I$, add an edge~$\{U_1,U_2\}$ to~$M$, where $U_1$ is the clique in~$G_1$ that contains~$v$ and $U_2$~is the clique in~$G_2$ that contains~$v$.
  To verify that~$M$ is a matching in~$H$, observe that each clique of~$G_1$ or~$G_2$ can contain only one vertex of~$I$.
  Therefore, the edges in~$M$ are pairwise disjoint and $|M|=|I|=k$.

  Second, let $M$~be a matching of size~$k$ in~$H$.
  We construct an independent set~$I$ for~$G$ with $|I|=k$ as follows: for each edge~$\{U_1,U_2\}\in M$, include an arbitrary vertex contained in both~$U_1$
  and~$U_2$ in~$I$.
  Since each clique of~$G_1$ or~$G_2$ is incident to at most one edge of~$M$, we have chosen at most one vertex per clique of~$G_1$ and~$G_2$, respectively.
  Hence, $I$~is an independent set. Furthermore, this implies that we did not choose a vertex twice, so~$|I|=|M|=k$.
  This completes the proof of the claim.

  It remains to prove the running time. Note that we can verify in linear time that a graph is a cluster graph.  Moreover, its connected components can be listed in linear time using depth first search. Hence, also the bipartite graph~$H$ is computable in linear time. Moreover, by construction, the graph~$H$~has at most $n$~edges, and, therefore, we can compute a maximum matching in~$H$ in $O(n^{1.5})$~time using the algorithm of Hopcroft and Karp~\citep[Theorem~16.4]{Sch03}.
\end{proof}

\noindent\looseness=-1 We point out that \autoref{lem:clusolve} generalizes to finding a maximum-\emph{weight} independent set if the vertices in the input cluster graphs have weights; to this end, we compute a \emph{weighted} maximum matching in the auxiliary bipartite graph, where each edge is assigned the maximum weight of the vertices occurring in both clusters it connects.

\newcommand{\POITS}{\textsc{3-SAT}}
\bigskip\noindent
To prove the second part of \autoref{thm:dichotomy}, we will employ a reduction from \POITS{}.

\decprob{\POITS}
{A Boolean formula $\phi$ in conjunctive normal form with at most three variables per clause.}
{Does $\phi$ have a satisfying assignment?}

\noindent In fact, we will see two very similar reductions from \POITS{} to \CIS{}, where the second is an extension of the first. This has the following benefits. The first reduction is an adaption of a simple NP-hardness reduction by \citet{GJS76} and, while not sufficient to show \autoref{thm:dichotomy}, it has properties that we will exploit to exclude polynomial-size problem kernels for \JIntSel{} in \autoref{sec:kernelization}:

\begin{lemma}\label{lem:mcis-NPhard}
  In polynomial time, a \POITS{} instance~$\phi$ can be reduced to a \CIS{} instance~$(G_1,G_2,k)$ such that 
  \begin{enumerate}[i)]
  \item $G_1$~consists of pairwise disjoint paths of length at most two and
  \item   $G_2$~consists of $k$ connected components, each of which is a triangle or an edge.
  \end{enumerate}
 Moreover, $k$ is proportional to the number of clauses in~$\phi$.
\end{lemma}

\noindent This first reduction, which we will use to prove
\autoref{lem:mcis-NPhard}, can then be modified to prove the following
lemma, which we will exploit to show \autoref{thm:dichotomy}. Note that,
in comparison to \autoref{lem:mcis-NPhard}, the following lemma
restricts $G_2$~to consist only of isolated edges and vertices.  A \pt{}
many-one reduction that \emph{additionally} ensures $G_2$~to have
$k$~connected components, like in \autoref{lem:mcis-NPhard}, does not
exist unless P${}={}$NP: in this case, \CIS{} is \pt{} solvable using
\textsc{2-Sat}.

\begin{lemma}\label{lem:mcis-NPharder}
  In polynomial time, a \POITS{} instance~$\phi$ can be reduced to a \CIS{} instance~$(G_1,G_2,k)$ such that 
  \begin{enumerate}[i)]
  \item   $G_1$~consists of pairwise disjoint paths of length at most two,
    \item $G_2$ consists only of isolated edges and vertices, and
    \item $G_1\cup G_2$ has chromatic index three.
    \end{enumerate}
Moreover, $k$ is proportional to the number of clauses in~$\phi$.
\end{lemma}

\begin{proof}[Proof of Lemmas \ref{lem:mcis-NPhard} and \ref{lem:mcis-NPharder}] 
  We first show \autoref{lem:mcis-NPhard}.
  To this end, we transform a \POITS{} formula $\phi$ into three
  graphs~$G_1', G_2'$, and~$G_3'$. The ``three graphs approach'' will
  easily allow us to show that the edge-wise union~$G_1:=G_1'\cup
  G_3'$ consists of pairwise disjoint paths of length at most two and
  that~$G_2:=G_2'$ consists of~$k$ pairwise disjoint triangles and~edges.

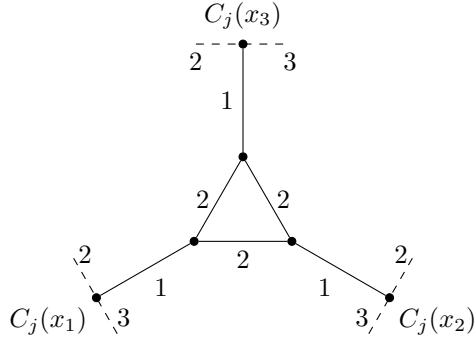
\begin{figure}
  \centering
  
  \begin{tikzpicture}[x=0.75cm,y=0.75cm]
    \node (i2) at (210:1) [vertex] {};
    \node (i3) at (330:1) [vertex] {};

    \node (i1) at (90:1) [vertex] {};

    \node (o2) at (210:3)[vertex] {};
    \node (o3) at (330:3)[vertex] {};
    \node (o1) at (90:3) [vertex] {};
    \node (x1a) at ($(180:1)+(o1)$) {};
    \node (x1b) at ($(0:1)+(o1)$) {};
    \draw[dashed] (x1a)--(x1b) node[midway,above=0.1]  {$C_j(x_3)$}
    node[at end, below] {$3$}
    node[at start, below] {$2$} ;

    \node (x2b) at ($(300:1)+(o2)$) {};
    \node (x2a) at ($(120:1)+(o2)$) {};
    \draw[dashed] (x2a)--(x2b) node[midway,below left] {$C_j(x_1)$}
    node[near start, above=0.1] {$2$}
    node[near end, right] {$3$};

    \node (x3a) at ($(60:1)+(o3)$) {};
    \node (x3b) at ($(240:1)+(o3)$) {};
    \draw[dashed] (x3a)--(x3b) node[midway, below right] {$C_j(x_2)$}
    node[near start, above=0.1] {$2$}
    node[near end, left] {$3$};

    \draw (i1)--(i2) node[midway, left] {$2$};
    \draw (i2)--(i3) node[midway, below] {$2$};
    \draw (i3)--(i1) node[midway, right] {$2$};

    \draw (i1)--(o1) node[midway, left] {$1$} ;
    \draw (i2)--(o2) node[midway, below right] {$1$};
    \draw (i3)--(o3) node[midway, below left] {$1$};
  \end{tikzpicture}

  \caption{Solid edges form the gadget for a clause~$C_j$ containing the variables~$x_1$, $x_2$, and~$x_3$. With each variable~$x_i$, we associate a leaf vertex~$C_j(x_i)$, which will be merged with either a T- or an F-vertex in the gadget for variable~$x_i$ (\autoref{fig:vargad}), depending on whether $C_j$ contains the variable~$x_i$ negated or not. The dashed edges represent edges of the variable gadgets. Edges labeled by a number~$\ell$ belong to the graph~$G_\ell'$.}
  \label{fig:simpleclogad}
\end{figure}

  We now give the details of the construction. Let~$\phi$ be a
  formula in conjunctive normal form with the clauses $C_1,\dots,C_m$,
  each of which contains at most three variables from the variable set
  $\{x_1,\hdots,x_n\}$. For a variable~$x_i$, let~$m_i$ denote the
  number of clauses in $\phi$ that contain~$x_i$.  For each
  clause~$C_j$ in~$\phi$, create the gadget shown in
  \autoref{fig:simpleclogad}, where an edge labeled~$\ell\in\{1,2,3\}$
  belongs to~$G_\ell'$. We call the non-triangle vertices \emph{leaves}
  and the non-triangle edges \emph{antennas}. With each variable~$x_i$
  in~$C_j$, we associate a leaf vertex~$C_j(x_i)$.

 For each variable~$x_i$, create a cycle gadget~$X_i$ (as illustrated in \autoref{fig:vargad}), with $2m_i$~edges, alternatingly labeled~$2$ and~$3$, and $2m_i$~vertices, which we alternatingly call \emph{T-vertex} and \emph{F-vertex}. %

  For each variable~$x_i$ of a clause~$C_j$, we merge the leaf
  vertex~$C_j(x_i)$ with a vertex of the variable gadget~$X_i$ for the
  variable~$x_i$ as follows: if~$x_i$ appears non-negated in~$C_j$,
  then we identify $C_j(x_i)$ with an F-vertex of $X_i$, otherwise, we
  identify~$C_j(x_i)$ with a T-vertex of $X_i$. Since $X_i$ has
  $m_i$~pairwise disjoint edges with label 3, each of which has one F-
  and one T-vertex, we can realize all these connections such that no
  edge with label 3 shares vertices with more than one antenna. The
  connections are illustrated by dashed lines in
  Figures~\ref{fig:simpleclogad} and~\ref{fig:vargad}.  We set
  $G_1:=G_1'\cup G_3'$ and $G_2:=G_2'$ and output the \CIS{} instance
  $(G_1, G_2, k)$, where $k:=m+\sum_{i=1}^{n}m_i$.

  It remains to show that the construction is correct and that~$G_1$
  and~$G_2$ satisfy the required properties. Indeed, $G_1$~consists of
  pairwise disjoint paths of lengths at most two, since it only
  contains the edges labeled~1 or~3, of which each family forms a
  matching, and no edge with label~3 is ever connected to two edges
  labeled~1 and vice versa (only antennas are labeled~1).  Moreover,
  $G_2$~consists of $k$~isolated triangles and edges (labeled~2): it
  contains one isolated triangle for each of the $m$~clauses and
  $m_i$~isolated edges for each variable~$x_i$. It remains to
  establish the correctness of the reduction by showing
  that~$\phi$~is satisfiable if and only if $G_1\cup G_2$~has an independent
  set of size~$k$.

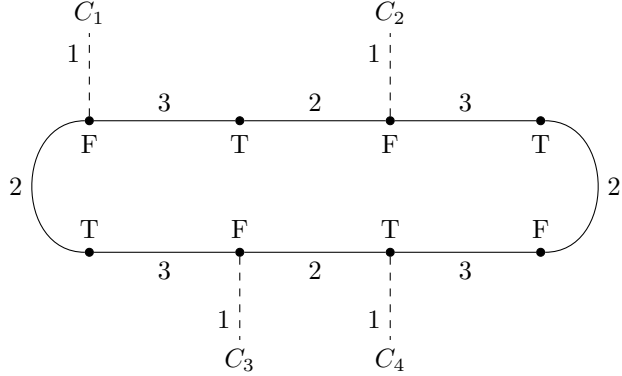
\begin{figure}
  \centering
  \begin{tikzpicture}[x=2cm, y=1.75cm]
    \node (l1) at (0,0) [vertex, label={above:T}] {};
    \node (l2) at (1,0) [vertex, label={above:F}] {};
    \node (l3) at (2,0) [vertex, label={above:T}] {};
    \node (l4) at (3,0) [vertex, label={above:F}] {};

    \node (u1) at (0,1) [vertex, label={below:F}] {};
    \node (u2) at (1,1) [vertex, label={below:T}] {};
    \node (u3) at (2,1) [vertex, label={below:F}] {};
    \node (u4) at (3,1) [vertex, label={below:T}] {};

    \draw (l1)--(l2) 
    node[pos=0.5, below] {$3$}
    
    --(l3) 
    node[pos=0.5, below] {$2$}
    
    --(l4)
    node[pos=0.5, below] {$3$}
    
    .. controls +(0:0.5)  and +(0:0.5) 
    .. (u4) 
    node [midway, right] {$2$}
    --(u3)  
    node[pos=0.5, above] {$3$}
    
    --(u2) 
    node[pos=0.5, above] {$2$}
    
    --(u1) 
    node[pos=0.5, above] {$3$}
    
    .. controls +(180:0.5) and +(180:0.5)
    ..(l1)
    node [midway, left] {$2$}
    
    --cycle;

    \draw[dashed] (u1) -- +(0,2/3) node[above] {$C_1$}
    node[near end, left] {$1$};

    \draw[dashed] (u3) -- +(0,2/3) node[above] {$C_2$}
    node[near end, left] {$1$};

    \draw[dashed] (l2) -- +(0,-2/3) node[below] {$C_3$}
    node[near end, left] {$1$};
    
    \draw[dashed] (l3) -- +(0,-2/3) node[below] {$C_4$}
    node[near end, left] {$1$};
  \end{tikzpicture}

  \caption{The solid edges form the gadget for a variable contained in  the clauses $C_1$, $C_2$, and $C_3$ non-negated, but in $C_4$ negated. Dashed edges belong to the clause gadgets of the clauses $C_1,\dots,C_4$. Edges labeled by a number~$\ell$ belong to the graph~$G_\ell'$.}
  \label{fig:vargad}
\end{figure}

First, let $I$~be an independent set for~$G_1\cup G_2$ that
  satisfies $|I|=m+\sum_{i=1}^{n}m_i=k$. %
  Note that, for each variable~$x_i$, the variable gadget of~$x_i$ is a
  cycle of $2m_i$~vertices. Clearly, $I$~contains at most half of
  these vertices. Moreover, $I$~contains at most one triangle vertex
  of each of our $m$~clause gadgets. Hence, $|I|\leq
  m+\sum_{i=1}^{n}m_i$ implies that $I$ contains
  a triangle vertex of each of the $m$~clause gadgets.
  Equivalently,
  \begin{itemize}
  \item for each variable gadget~$X_i$, either all T-vertices or all
    F-vertices are contained in~$I$, and
  \item for each clause gadget, one of its leaf vertices is \emph{not}
    contained in~$I$.
  \end{itemize}
  Equivalently, in each clause~$C_j$,
  we find at least one of the following situations:
  \begin{itemize}
  \item $C_j$~contains a positive literal~$x_i$ (then $C_j(x_i)$~is an
    F-vertex in~$X_i$) and $I$~contains all T-vertices of~$X_i$, or
  \item $C_j$~contains a negative literal~$\bar x_i$ (then $C_j(x_i)$
    is a T-vertex in~$X_i$) and $I$~contains all F-vertices of~$X_i$.
  \end{itemize}
  Therefore, on the one hand, setting a variable~$x_i$ to true if and
  only if $I$~contains the T-vertices of~$X_i$ yields a satisfying
  assignment for~$\phi$. 

  \looseness=-1 On the other hand, if we have a satisfying
  assignment for~$\phi$, putting into~$I$ all T-vertices of~$X_i$
  if $x_i$ is true and all F-vertices otherwise allows us to
  choose~$I$ so that it contains a triangle vertex of each clause gadgets and, thus, so that~$|I|\geq k$.

\begin{figure}[t]
  \centering
  
  \begin{tikzpicture}[x=0.9cm,y=0.9cm]
    \node (i2) at (210:1) [vertex] {};
    \node (i3) at (330:1) [vertex] {};
    \node (i1) at (90:1) [vertex] {};

    \node (i12a) at ($(i1)!1/3!(i2)$) [vertex] {};
    \node (i12b) at ($(i1)!2/3!(i2)$) [vertex] {};

    \node (i23a) at ($(i2)!1/3!(i3)$) [vertex] {};
    \node (i23b) at ($(i2)!2/3!(i3)$) [vertex] {};

    \node (i31a) at ($(i3)!1/3!(i1)$) [vertex] {};
    \node (i31b) at ($(i3)!2/3!(i1)$) [vertex] {};

    \node (o2) at (210:4) [vertex] {};
    \node (o3) at (330:4) [vertex] {};
    \node (o1) at (90:4) [vertex] {};

    \node (o1a) at ($(i1)!1/5!(o1)$) [vertex] {};
    \node (o1b) at ($(i1)!2/5!(o1)$) [vertex] {};
    \node (o1c) at ($(i1)!3/5!(o1)$) [vertex] {};
    \node (o1d) at ($(i1)!4/5!(o1)$) [vertex] {};

    \node (o2a) at ($(i2)!1/5!(o2)$) [vertex] {};
    \node (o2b) at ($(i2)!2/5!(o2)$) [vertex] {};
    \node (o2c) at ($(i2)!3/5!(o2)$) [vertex] {};
    \node (o2d) at ($(i2)!4/5!(o2)$) [vertex] {};

    \node (o3a) at ($(i3)!1/5!(o3)$) [vertex] {};
    \node (o3b) at ($(i3)!2/5!(o3)$) [vertex] {};
    \node (o3c) at ($(i3)!3/5!(o3)$) [vertex] {};
    \node (o3d) at ($(i3)!4/5!(o3)$) [vertex] {};

    \node (x1a) at ($(180:1)+(o1)$) {};
    \node (x1b) at ($(0:1)+(o1)$) {};
    \draw[dashed] (x1a)--(x1b) node[midway,above=0.1]  {$C_j(x_3)$}
    node[at end, below] {$3$}
    node[at start, below] {$2$} ;

    \node (x2b) at ($(300:1)+(o2)$) {};
    \node (x2a) at ($(120:1)+(o2)$) {};
    \draw[dashed] (x2a)--(x2b) node[midway,below left] {$C_j(x_1)$}
    node[near start, above=0.1] {$2$}
    node[near end, right] {$3$};

    \node (x3a) at ($(60:1)+(o3)$) {};
    \node (x3b) at ($(240:1)+(o3)$) {};
    \draw[dashed] (x3a)--(x3b) node[midway, below right] {$C_j(x_2)$}
    node[near start, above=0.1] {$2$}
    node[near end, left] {$3$};

    \draw (i1)--(i12a) node[midway,  left] {$1$}
    --(i12b) node[midway,  left] {$2$}
    --(i2) node[midway,  left] {$3$};

    \draw (i2)--(i23a) node[midway, below] {$1$}
    --(i23b) node[midway, below] {$2$}
    --(i3) node[midway, below] {$3$};

    \draw (i3)--(i31a) node[midway,  right] {$1$}
    --(i31b) node[midway,  right] {$2$}
    --(i1) node[midway,  right] {$3$};

    \draw (i1)--(o1a) node[midway, left] {$2$}
    --(o1b) node[midway, left] {$1$}
    --(o1c) node[midway, left] {$3$}
    --(o1d) node[midway, left] {$2$}
    --(o1) node[midway, left] {$1$};

    \draw (i2)--(o2a) node[midway, above left] {$2$}
    --(o2b) node[midway, above left] {$1$}
    --(o2c) node[midway, above left] {$3$}
    --(o2d) node[midway, above left] {$2$}
    --(o2) node[midway, above left] {$1$};

    \draw (i3)--(o3a) node[midway, above right] {$2$}
    --(o3b) node[midway, above right] {$1$}
    --(o3c) node[midway, above right] {$3$}
    --(o3d) node[midway, above right] {$2$}
    --(o3) node[midway, above right] {$1$};

  \end{tikzpicture}

  \caption{Advanced gadget for a clause~$C_j$ containing the
    variables~$x_1$, $x_2$, and~$x_3$. With each variable~$x_i$, we
    associate a leaf vertex~$C_j(x_i)$, which is merged with either
    a T- or an F-vertex in the gadget for variable~$x_i$
    (\autoref{fig:vargad}), depending on whether $C_j$ contains the
    variable~$x_i$ negated or not. The dashed edges represent edges
    of the variable gadgets. Edges labeled by a number~$\ell$ belong
    to the graph~$G_\ell'$.}
  \label{fig:clogad}
\end{figure}
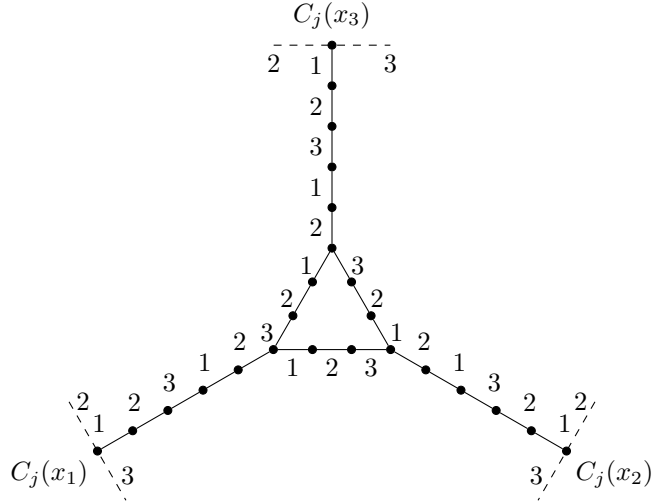

  \bigskip\noindent We can now easily turn this reduction into a
  reduction that also proves \autoref{lem:mcis-NPharder}. To this end,
  note that subdividing an edge of a graph twice increases the size of
  the graph's maximum independent set by exactly one. Thus, instead of
  using the clause gadget in \autoref{fig:simpleclogad}, we can use
  the gadget shown in \autoref{fig:clogad} and ask for an independent
  set of size~$k:=10m+\sum_{i=1}^{n}m_i$: indeed, the gadget in
  \autoref{fig:clogad} is obtained from the simpler one in
  \autoref{fig:simpleclogad} by subdividing each triangle edge
  twice and each antenna four times. Thus, we increase the maximum
  independent set size by $3+3\cdot 2=9$ per clause gadget and, hence,
  ask for $k:=10m+\sum_{i=1}^{n}m_i$ instead of $m+\sum_{i=1}^{n}m_i$.

  The benefit of replacing the gadget in \autoref{fig:simpleclogad} by
  the gadget in \autoref{fig:clogad} is that~$G_2$ now consists only of
  isolated edges and vertices instead of triangles. Moreover, the
  resulting graph $G_1\cup G_2=G_1'\cup G_2'\cup G_3'$ has chromatic
  index three, where the edge labels yield a proper edge coloring.
  Thus, using \autoref{fig:clogad}, we proved
  \autoref{lem:mcis-NPharder}.  %
\end{proof}

\noindent Using \autoref{lem:mcis-NPharder}, it is now easy to prove \autoref{thm:dichotomy}.

\begin{proof}[Proof of \autoref{thm:dichotomy}]
  Statement (i) immediately follows from \autoref{lem:clusolve}. It remains to show (ii). 
  To this end, observe that, without loss of generality, $\mathcal C_1$ contains not only cluster graphs. Therefore, it contains a graph that has a~$P_3$ as induced subgraph.  Since $\mathcal C_1$ is closed under induced subgraphs and disjoint unions, it follows that $\mathcal C_1$ contains all graphs that consist of pairwise disjoint paths of length at most two. With the same argument and exploiting that $\mathcal C_2$ contains at least one graph with an edge, we obtain that $\mathcal C_2$ contains all graphs consisting of isolated vertices and edges. Hence, NP-hardness follows from \autoref{lem:mcis-NPharder}.
\end{proof}

\noindent Concluding this section,  we derive further hardness results for \JIntSel{} from \autoref{lem:mcis-NPhard}.

\begin{corollary}\label{cor:NPhwithKColors}
  Even when restricted to instances 
  \begin{itemize}
  \item with an input graph that consists of disjoint paths of length at most two
  \item and that ask for a colorful independent set of size~$k$ with $k$~being equal to the number~$\gamma$ of input colors,
  \end{itemize}
  \JIntSel{} remains
  \begin{enumerate}[i)]
  \item\label{eins} NP-hard, and
  \item\label{zwei} cannot be solved in $2^{o(k)}\cdot n^{O(1)}$~time unless the Exponential Time Hypothesis fails.\footnote{The Exponential Time Hypothesis basically states that there is no~$2^{o(n)}$-time algorithm for $n$-variable \POITS{} \citep{journal-IPZ01,LMS11}.}
  \end{enumerate}
\end{corollary}

\noindent Herein, \eqref{eins} simply translates into our colored
model
the statement of \autoref{lem:mcis-NPhard} that \CIS{} remains
NP-hard even if one input graph is a cluster graph with $k$~connected
components. Moreover, (ii) follows by combining a result of
\citet{journal-IPZ01} with the fact that $k$~is proportional to the
number of clauses in the input \POITS{} formula
(\autoref{lem:mcis-NPhard}).

\section{Job Interval Selection}\label{sec:jisp}

In this section, we investigate the parameterized complexity of \JIntSel{}. As warm-up for working with the colored model, \autoref{sec:searchtrees} first gives a simple search tree algorithm that solves \JIntSel{} in linear time if the sought solution size~$k$ and the maximum number~$\maxcolcli$ of colors in any maximal clique of~$G$ are constant. %

\autoref{sec:DP} then proceeds with a reformulation of the \fp{} algorithm of \citet{HK06} with respect to a structural parameter into our colored model, which makes it easy for us to generalize the algorithm to \MCIS{} and also to show that the problem is linear-time solvable if only $k$~is constant (as opposed to requiring both $k$ and $\maxcolcli$ being constant). However, the space requirements as well as the running time exponentially depend on~$k$.

We conclude our findings for \JIntSel{} in \autoref{sec:kernelization} by showing that the problem has no polynomial-size problem kernel in general, but on proper interval graphs.

\subsection{A Simple Search Tree Algorithm}
\label{sec:searchtrees}

As a warm-up for working with the colored formulation of \JIntSel{}, this section presents a simple search tree algorithm leading to the following theorem:

\begin{theorem}\label{thm:searchtree}
\JIntSel{} is solvable in $O(\maxcolcli^k\cdot n)$ time, where $\maxcolcli$ is the maximum number of colors occurring in any maximal clique.
\end{theorem}
\noindent
\looseness=-1 Only for $\Gamma<6$ the worst-case running time of \autoref{thm:searchtree} can compete with our generalizations of the dynamic program of \citet{HK06} in \autoref{sec:DP} (\autoref{thm:colorcoding}). However, as opposed to the dynamic programs presented in \autoref{sec:DP}, the space requirements of the search tree algorithm are polynomial.

\bigskip\noindent\looseness=-1 The first ingredient in our search tree algorithm is
the following lemma, which shows that a search tree algorithm only has
to consider the ``first'' intervals of the interval graph for inclusion into
an optimal solution. This is illustrated in \autoref{fig:Kset}.

\begin{figure}
  \centering
  \begin{tikzpicture}

    \draw[interval,lightgray] (2.5,0)--(3.5,0);
    \draw[interval,lightgray] (3.1,0.3)--(4.5,0.3);
    \draw[interval] (0,0)--(2,0);
    \draw[interval] (2,-0.3)--(4,-0.3);
    \draw[interval] (-1,0.3)--(2.5,0.3);
    \draw[dashed] (2,-0.75)--(2,0.75);
  \end{tikzpicture}
  \caption{The set $K$~of intervals that start no later than any interval in~$G$ ends are drawn black. The other intervals are drawn gray.}
  \label{fig:Kset}
\end{figure}
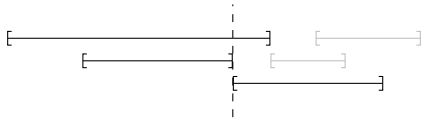

\begin{lemma}\label{lem:firstclique} %
  Let $K$~be the set of intervals that start no later than any interval
  in~$G$ ends. Then, there is a maximum colorful independent set that
  contains exactly one vertex of~$K$.
\end{lemma}

\begin{proof}
  Let $I$~be a maximum colorful independent set for~$G$ with $I\cap
  K=\emptyset$ and let $v^*$ be the interval in~$G$ that ends
  first. Obviously, $v^*\in K$ and any interval $v\in I$
  intersecting~$v^*$ is in~$K$. Hence, since $I\cap K=\emptyset$,
  $I$~contains no interval intersecting~$v^*$.  It follows that $I$
  contains a vertex~$w$ such that $\col(w)=\col(v^*)$, otherwise
  $I\cup\{v^*\}$ would be a larger colorful independent set. Now,
  $I'=(I\setminus\{w\})\cup\{v^*\}$ is a colorful independent set
  for~$G$ with $|I'|=|I|$ and~$v^*\in I'\cap K$.

  Finally, note that $I$~cannot contain more than one vertex of~$K$
  since the intervals in~$K$ pairwise intersect.  
\end{proof}

\noindent The second ingredient in our search tree algorithm is the
following lemma, which shows that knowing the color of the interval in
$K$ that is to be included in an optimal solution is sufficient to
choose an optimal interval from~$K$ into a maximum colorful independent
set.

\begin{lemma}\label{lem:firstend}
  Let $K$~be the set of intervals that start no later than any
  interval in~$G$ ends. Moreover, assume that there is
  a maximum
  colorful independent set containing an interval of color~$c$ from~$K$.

  Then, there is a maximum colorful independent set that contains the
  interval of color~$c$ from~$K$ that ends first.
\end{lemma}

\begin{proof}
  Let $I$~be a maximum colorful independent set, let~$v\in K\cap I$ and let~$\col(v)=c$. Moreover, let $v^*$ be the interval in~$K$ with $\col(v^*)=c$ that ends first. By \autoref{lem:firstclique}, $I$~contains at most one interval of~$K$.
  Then, since $v$ intersects all intervals that intersect~$v^*$, we know that $I'=(I\setminus\{v\})\cup\{v^*\}$ is a colorful independent set with $|I'|=|I|$.  
\end{proof}

\noindent Using \autoref{lem:firstclique} and \autoref{lem:firstend}, it is easy to prove \autoref{thm:searchtree}.

\begin{proof}[Proof of \autoref{thm:searchtree}]
  The algorithm works as follows. First, find the set~$K$ of intervals
  that start no later than any interval in~$G$ ends. Let
  $C:=\bigcup_{v\in K}\col(v)$ be the set of colors occurring
  in~$K$. Note that these computations can be executed in
  $O(n)$~time. Since the intervals in~$K$ form a maximal clique, it
  follows that $|C|\leq\maxcolcli$. By \autoref{lem:firstclique} and
  \autoref{lem:firstend}, it is now sufficient, for each color~$c\in C$
  and the first-ending interval~$v$ with $\col(v)=c$, to try
  choosing~$v$ for inclusion into the solution and to try recursively
  finding a colorful independent set of size~$k-1$ in the interval
  graph~$G$ without vertices having color~$c$ or intersecting~$v$ (that
  is, starting after $v$~ends).

The recursion depth is bounded by~$k$, each recursion step causes at most $\maxcolcli$~new recursion steps, and each recursion step requires $O(n)$~time, yielding a total running time of $O(\maxcolcli^k\cdot n)$. 
\end{proof}

\subsection{Generalizations of \citet{HK06}'s Dynamic Program}
\label{sec:DP}

In this section, we first present the dynamic program for \JIntSel{} by \citet{HK06} in terms of our colored model. Based on this presentation, we show modifications in order to lower its space requirements, we generalize it to \MColIS{} and, finally, transform it into a \fp{} algorithm with respect to the parameter~$k$.

It is easy to see that the dynamic programs in this section can be straightforwardly generalized to the problem variant where each interval has assigned a weight and we search for a colorful independent set of maximum weight, rather than  of maximum size.

\paragraph{\boldmath Dynamic Program for Parameter ``Number~$\gamma$ of Colors''.}
Let $(G,k)$~be an instance of \JIntSel{}, where~$G$ is given in
$c$-compact representation for minimum~$c$. For $i\in [c+1]$ and
$C\subseteq[\gamma]$, we use~$\entry{i,C}$ to denote the size of a
maximum colorful independent set in~$G$ that uses only intervals whose
start point is at least~$i$ and whose color is in~$C$.  Obviously,
for~$i=c+1$ and any $C\subseteq[\gamma]$, we have
$\entry{i,C}=0$. Knowing~$\entry{i,C}$ for some $i\in[c+1]$ and
all~$C\subseteq[\gamma]$, we can easily compute $\entry{i-1,C}$ for
all~$C\subseteq[\gamma]$, since there are only two cases:
\begin{noindlist}
\item There is a maximum independent set of intervals with start point at least~$i-1$ and colors belonging to~$C$ that contains an interval~$v$ with $v_s=i-1$. Then, $\entry{i-1,C}=1+\entry{{v_e+1},C\setminus\{\col(v)\}}$.
\item Otherwise, $\entry{i-1,C}=\entry{i,C}$.
\end{noindlist}
It follows that we can compute the size $\entry{1,[\gamma]}$ of a maximum colorful independent set in~$G$ using the recurrence
\begin{align*}
&\entry{i-1,C}\tag{DP-$\gamma$}\label{DP-gamma}\\&=\max
  \begin{cases}
    T[i,C],\\
    1+{}\smashoperator{\mathop{\smashoperator{\max\limits_{v\in V,v_s=i-1}}}\limits_{\col(v)\in C}}T[v_e+1,C\setminus\{\col(v)\}].\hspace{-1em}
  \end{cases}
\end{align*}
In this way, we obtain an alternative formulation of the dynamic
program of \citet{HK06} using colored interval graphs instead of
a geometric formulation. We can evaluate recurrence~\eqref{DP-gamma}
in~$O(2^\gamma n)$~time by iterating over the intervals in~$G$ in order
of decreasing start points and, for each interval, iterating over all
subsets of~$[\gamma]$. In this way, we first handle all intervals with
start point~$c$, then with~$c-1$ and so on, so that we compute the
table entries for decreasing start points~$i\in[c+1]$. Herein, the
$c$-compact representation not only ensures that the intervals are
sorted by their start points, but also that, for each~$i\in[c]$, some
interval starts in~$i$ and, therefore, that the table entry~$T[i,C]$
indeed gets filled for all~$C\subseteq[\gamma]$. This algorithm yields an
alternative proof for a result by \citet{HK06}:

\begin{proposition}[\citet{HK06}]\label{lem:gamma}
  \JIntSel{} is solvable in $O(2^\gamma\cdot n)$~time.
\end{proposition}

\paragraph{\boldmath Dynamic Program for Parameter ``Maximum Number~$Q$ of Live Colors''.}
\citet{HK06} improved recurrence \eqref{DP-gamma} from using
the parameter~$\gamma$ to the structural parameter~$Q\leq \gamma$, which
is defined as follows:

\begin{definition}\label{def-Q} Let $G$~be an interval graph given in
  $c$-compact representation and with vertex colors in~$[\gamma]$. For
  each $i\in[c+1]$, let
  \begin{description}
  \item[{$L_i\subseteq[\gamma]$}] be the set of colors that appear on
    intervals with start point at most~$i$ (note that
    $L_{c+1}=[\gamma]$), and
  \item[{$R_i\subseteq[\gamma]$}] be the set of colors that appear on
    intervals with start point at least~$i$ (note that
    $R_{c+1}=\emptyset$).
  \end{description}
  Then $Q:=\max_{i\in[c+1]}|L_i\cap R_i|$ is the maximum number of
  \emph{live colors}.  That is, a color~$c$ is \emph{live} at a
  point~$i$ if there is an interval with color~$c$ that starts no
  later than~$i$ as well as an interval with color~$c$ that starts no
  earlier than~$i$.
\end{definition}

Using this definition, we first observe that,
when searching for a maximum colorful independent set containing only intervals with start point at least~$i$, it is safe to allow
this independent set to contain all colors of~$\bar L_i:=[\gamma]\setminus L_i$: this is because an interval with start point before~$i$ cannot have a color in~$\bar L_i$. Hence, we are only interested in the values~$\entry{i,C}$ for $i\in[c+1]$ and $\bar L_i\subseteq C\subseteq [\gamma]$.
Second, a colorful independent set that only contains intervals with start point at least~$i$ only contains intervals of color~$R_i$. Therefore, it is safe to allow only colors contained in~$R_i$ and we see that we are only interested in the values~$\entry{i,C}$ for $i\in[n+1]$ and $\bar L_i\subseteq C\subseteq R_i$.
There are at most~$2^Q$~such subsets, since for each~$C$ with $\bar L_i\subseteq C\subseteq R_i$, we have $C\setminus\bar L_i\subseteq L_i\cap R_i$.

Exploiting these observations in \eqref{DP-gamma}, we can compute $T[i-1,C]$ for all~$C$ with~$\bar L_{i-1}\subseteq C \subseteq R_{i-1}$ as
\begin{align*}
  &\entry{i-1,C}\tag{DP-$Q$}\label{DP-Q}\\&=\max
  \begin{cases}
    T[i,(C\cup\bar L_i)\cap R_i],\\
    1+{}\smashoperator{\mathop{\smashoperator{\max\limits_{v\in V,v_s=i-1}}}\limits_{\col(v)\in C}}T[v_e+1,(C\cup\bar L_{v_e+1})\cap (R_{v_e+1}\setminus\{\col(v)\})].
  \end{cases}
\end{align*}
As we have not changed the semantics of a table entry compared to \eqref{DP-gamma}, the size of a maximum colorful independent set in~$G$ is, as before, $\entry{1,[\gamma]}$. Hence, the improved dynamic program of \citet{HK06} also works in our colored model:

\begin{proposition}[\citet{HK06}]
  \JIntSel{} is solvable in $O(2^Q\cdot n)$~time, where $Q$~is the maximum number of live colors  as defined in \autoref{def-Q}.
\end{proposition}

\paragraph{Improving the Space Complexity.} Having stated the dynamic programs of \citet{HK06} in terms of our colored model, we now build upon these algorithms.  Obviously, the dynamic programming table of recurrence \eqref{DP-Q} has $2^Q\cdot (c+1)$~entries. We improve it to~$2^Q \cdot(\ell+2)$, where $\ell$ is the length of the longest interval in the input interval graph. That is, if $Q$ and $\ell$ are constant, we can solve arbitrarily large input instances using a constant-size dynamic programming table. %
Note that, even if $\ell$~is not bounded by a constant, we have $\ell\leq c-1$, and therefore $2^Q\cdot(\ell+2)\leq 2^Q\cdot(c+1)$, since the input instance is given in a $c$-compact representation.

The improvement of space complexity is based on a simple observation: when computing $\entry{i-1,C}$ in \eqref{DP-Q}, there is a largest possible~$i'>i-1$ and some color set~$C'$ for which we access $\entry{i',C'}$. By definition of~$\entry{}$, $i'=v_e+1$ for some interval~$v$ with start point~$v_s=i-1$. We have $i'-1=v_e\leq v_s+\ell=i-1+\ell$, and, hence, $i'\leq i+\ell$. It follows that we only need $2^Q(\ell+2)$ table entries, since the entry~$\entry{i-1,C}$ does not need the value~$\entry{i+\ell+1,C}$ and can therefore reuse the space previously occupied by~$\entry{i+\ell+1,C}$. This we simply achieve by storing~$\entry{i,C}$ for~$i\in [c+1]$ and~$C\subseteq [\gamma]$ in a table~$\entryp{i \bmod (\ell +2),C}$ that has only~$\ell+2$ entries in the first coordinate. Having shrunken the dynamic programming table in this way, we obtain the following lemma:

\begin{proposition}\label{lem:ell-Q}
  \JIntSel{} is solvable in $O(2^Q\cdot n)$~time and
  $O(2^Q\ell+\gamma c)$~space when the input graph is given in
  $c$-compact representation, $\ell$ is the maximum interval length,
  and $Q$~is the maximum number of live colors as defined in
  \autoref{def-Q}.
\end{proposition}

\noindent Herein, $O(2^Q\ell$)~space is used by the dynamic
programming table and $O(\gamma c)$~space is used to hold the
sets~$L_i$ and~$R_i$ from \autoref{def-Q}, which we used to speed up
the dynamic programming.

\paragraph{Generalization to \MColIS{}.}  We now generalize \eqref{DP-Q} to \MColIS{}. That is, vertices are now allowed to have multiple colors instead of just one and we search for a maximum independent set that is colorful in the sense that no pair of vertices may have common colors. The algorithm for \MColIS{} will allow us to solve \MCIS{} in \autoref{fpt-mcis}.

Due to the formulation of \eqref{DP-Q} in our colored model, the generalization to \MColIS{} turns out to be easy. For $i\in [c+1]$ and $C\subseteq [\gamma]$, we use $\entry{i,C}$ to denote the size of a maximum colorful independent set in~$G$ that uses only intervals with start point at least~$i$ and whose colors are a subset of~$C$.

Completely analogously to \JIntSel{}, we can compute the size~$\entry{1,[\gamma]}$ of a maximum colorful independent set by computing $\entry{i-1,C}$ for each~$i\in[c+1]$ and all color sets~$C$ with $\bar L_{i-1}\subseteq C\subseteq R_{i-1}$ as
\begin{align*}
  &\entry{i-1,C}\tag{DP-$Q$*}\label{DP-Qp}\\&=\max
  \begin{cases}
T[i,(C\cup\bar L_i)\cap R_i],\\
1+{}\smashoperator{\mathop{\smashoperator{\max\limits_{v\in V,v_s=i-1}}}\limits_{\col(v)\subseteq C}}T[v_e+1,(C\cup\bar L_{v_e+1})\cap (R_{v_e+1}\setminus\col(v))].
  \end{cases}
\end{align*}
The improvement of the space complexity demonstrated for \JIntSel{} also works here. Hence, we can merge Propositions~\ref{lem:gamma}--\ref{lem:ell-Q} into the following theorem:

\begin{theorem}\label{thm:mcolis-fpt}
  Given an interval graph with $\gamma$~colors and maximum interval
  length~$\ell$ in $c$-compact interval representation, \MColIS{} is
  solvable in $O(2^Q\cdot n)$~time and $O(2^Q\ell+\gamma
  c)$~space, where $Q$~is the maximum number of live colors as defined
  in \autoref{def-Q}.
\end{theorem}

\paragraph{\boldmath Algorithm for Parameter ``Solution Size~$k$''.}
\label{sec:JIS k fpt}
We now improve recurrence \eqref{DP-gamma} to a \fp{} algorithm for \JIntSel{} with respect to the parameter~${k\leq \gamma}$. Our first step is providing a randomized \fp{} algorithm for \JIntSel{}. The algorithm correctly answers if a no-instance of \JIntSel{} is given. In contrast, it rejects ``yes''-instances with a given error probability~$\varepsilon$. The randomized algorithm can be derandomized to show the following theorem:

\begin{theorem}\label{thm:colorcoding}
  \JIntSel{} can be solved
 with error probability~$\varepsilon$ in $O(5.5^k\cdot |\ln\varepsilon|\cdot n)$~time and $O(2^k\cdot\ell)$~space. The algorithm can be derandomized to deterministically solve \JIntSel{} in $O(12.8^{k}\cdot \gamma n)$~time.
\end{theorem}

\noindent Comparing this theorem with the hardness result in \autoref{cor:NPhwithKColors} from \autoref{sec:dichotomy}, the running time of the derandomized algorithm is optimal up to factors in the base. However, in practical applications, the randomized algorithm is probably preferable over the derandomized one, since the error probability can be chosen very low without increasing the running time significantly. 

 To prove \autoref{thm:colorcoding}, we use the color-coding technique by \citet{AYZ95} to reduce the number~$\gamma$ of colors in the given instance to~$k$. After that, recurrence \eqref{DP-gamma} can be evaluated in~$O(2^k\cdot n)$~time. Depending on whether we reduce the number of colors randomly or deterministically, this method will yield the first or the second running time.

\newcommand{\recoloring}{\delta}
\begin{proof}[Proof of \autoref{thm:colorcoding}]
  Let~$(G,\col,k)$ be an instance of \JIntSel.  In a first step, we
  assign each color in~$[\gamma]$ a color in~$[k]$ uniformly at
  random.  Let $\recoloring\colon[\gamma]\rightarrow[k]$ denote this recoloring
  and let $(G,\col',k)$ denote the resulting instance
  with~$\col'(v)=\recoloring(\col(v))$ for all vertices~$v$. Note that, in
  general, $\recoloring$~is not injective.  Then, we use~\eqref{DP-gamma} to
  compute a size-$k$ colorful independent set in the resulting
  instance. Since the resulting instance has only $k$~colors, this
  works in $O(2^k\cdot n)$~time.

  We now first analyze the probability that a colorful
  independent set for~$(G,\col,k)$ is also a colorful independent set
  for~$(G,\col',k)$ and vice versa. Then, we analyze how often we have
  to repeat the procedure of recoloring and computing
  recurrence~\eqref{DP-gamma} in order to achieve the low error
  probability~$\varepsilon$.

  First, assume that the recolored instance~$(G,\col',k)$ is a
  ``yes''-instance. Then, there is a colorful independent set~$I$ with
  $|I|\geq k$. The set~$I$ is a colorful independent set also for the
  original instance $(G,\col,k)$, since each color in $\col$ is mapped
  to only one color in~$\col'$. It follows that $(G,\col,k)$~is a
  ``yes''-instance.

  Now, assume that the original instance~$(G,\col,k)$ is a
  ``yes''-instance. We analyze the probability of the recolored instance
  $(G,\col',k)$ being a ``yes''-instance. Let~$I$ be a colorful
  independent set for~$(G,\col,k)$. The set~$I$ is a colorful
  independent set for~$(G,\col',k)$ if the vertices in~$I$ have
  pairwise distinct colors with respect to~$\col'$.  Since the
  vertices in~$I$ have pairwise distinct colors with respect
  to~$\col{}$ and we assign each color in~$[\gamma]$ a color in~$[k]$
  uniformly at random, the colors of the vertices of~$I$ with respect
  to~$\col'$ are also chosen uniformly at random and independently
  from each other.  Thus, the probability of $I$~being colorful with
  respect to~$\col'$ is~$p:=k!/k^k$: out of $k^k$~possible ways of
  coloring the $k$~vertices in~$I$ with $k$~colors, there are
  $k!$~ways of doing so in a colorful manner. Hence, the probability
  of $(G,\col',k)$ also being a ``yes''-instance, is~$p:=k!/k^k$.

  In order to lower the error probability of not finding a colorful
  independent set if it exists to~$\varepsilon$, we repeat the process
  of recoloring and running recurrence~\eqref{DP-gamma}
  $t(\varepsilon)$~times. That is, we want
\begin{align*}
  \hspace{3em}(1-p)^{t(\varepsilon)}&\leq\varepsilon.
\intertext{Exploiting that $1+x\leq e^x$ holds for all~$x\in\mathbb R$, the above inequality is satisfied by any number~$t(\varepsilon)$ of recoloring trials that satisfies}\
  \hspace{3em}e^{-p\cdot t(\varepsilon)}&\leq\varepsilon.
\end{align*}
Taking the logarithm on both sides and rearranging terms,
\begin{align*}
  \hspace{3em}t(\varepsilon)&\geq \ln\varepsilon\cdot \frac{1}{-p}=|\ln\varepsilon|\cdot \frac{k^k}{k!}.
\end{align*}
Using Stirling's lower bound for the factorial, one obtains
$k^k/k!\in O(e^k)$. To conclude the proof, it is now enough to put
together the observations that each run of recurrence~\eqref{DP-gamma}
with $k$~colors takes $O(2^k\cdot n)$~time and that we have to
repeat it only $t(\varepsilon)\in O(|\ln\varepsilon|\cdot
e^k)$~times to get an error probability of~$\varepsilon$. Thus, the
overall procedure takes $O(|\ln\varepsilon|\cdot (2e)^k\cdot
n)$~time.

\medskip\noindent We now derandomize the presented
algorithm: instead of repeatedly choosing random
recolorings~$\recoloring\colon[\gamma]\to[k]$, we deterministically enumerate
the recolorings according to a \emph{$k$-color coding scheme}
\citep{CLSZ07}: a $k$-color coding scheme~$\mathcal F$ is a set of
recolorings such that, for each subset~$C\subseteq[\gamma]$ with
$|C|=k$, there is a recoloring~$\recoloring\in\mathcal F$ such that the colors
in~$C$ will be mapped to pairwise distinct colors by~$\recoloring$. That is,
whatever colors a colorful independent set~$I$ of size~$k$ in~$G$
might have, there is one recoloring in~$\mathcal F$ such that $I$~is
colorful after recoloring. Thus, the dynamic
program~\eqref{DP-gamma} will find it.

A $k$-color coding scheme~$\mathcal F$ can be computed in
$O(6.4^k\cdot\gamma)$~time \citep{CLSZ07}. Moreover, it consists
of $O(6.4^k\cdot\gamma)$ colorings.  That is, in
$O(6.4^k\gamma\cdot 2^kn)$~time, we can run~\eqref{DP-gamma} for
each coloring in~$\mathcal F$, thus proving~(ii).
\end{proof}

\noindent Many algorithms that are based on the color-coding techniques
can be sped up using algebraic techniques~\citep{KW09}.  It would be
interesting to see whether they can also be used to speed up the running
time of \autoref{thm:colorcoding} (at least in the asymptotic sense).

Finally, note that the color-coding technique as used in
\autoref{thm:colorcoding} for \JIntSel{} could be applied to \MColIS{}
in the same way. However, the result will not be a \fp{} algorithm with
respect to the parameter ``solution size~$k$'', but with respect to the
total number of colors found in the lists of the solution vertices. This
number could potentially be much larger than~$k$ and even~$n$, thus
making such a \fp{} algorithm not particularly attractive for \MColIS{}.

\subsection{Polynomial-Time Preprocessing}
\label{sec:kernelization}

In this section, we first show that efficient and effective data reduction in form of polynomial-size problem kernels is most likely unfeasible for \JIntSel{}. Then, we show that it becomes feasible when we restrict the colored input graph to be a \emph{proper} interval graph.

\paragraph{Non-Existence of Polynomial-Sized Problem Kernels.}\label{sec:no-poly-kernel}
\newcommand{\shift}{\ensuremath{\mathrm{shift}}}
\newcommand{\flip}{\ensuremath{\mathrm{flip}}}
\looseness=-1 We show that \JIntSel{} is unlikely to admit problem kernels of polynomial size with respect to various parameters.
To this end, we employ the ``cross composition'' technique introduced by \citet{BJK14}.
A \emph{cross composition} is a polynomial-time algorithm that, given $t$~instances~$x_i$ with~$0\leq i <t$ of an \NP-hard starting problem~$A$, outputs an instance~$(y,k)$ of a parameterized problem~$B$ such that $k\in\poly(\max_{0\leq i< t}|x_i|+\log t)$ and $(y,k)$ is a ``yes''-instance for $B$ if and only if there is some $0\leq i< t$ with $x_i$~being a ``yes''-instance for~$A$. A theorem by \citet{BJK14} now states that if a problem~$B$ admits such a cross composition, then there is no polynomial-size problem kernel for~$B$ unless the polynomial hierarchy collapses to the third level, which is widely disbelieved. %

In the following, we present a cross composition for \JIntSel{} parameterized by the combination of the size~$\pw$ of a maximum clique and the number~$\cols$ of colors, yielding the following theorem:

\begin{theorem}\label{thm:uig-nopoly}
  Unless the polynomial hierarchy collapses, \JIntSel{} does not admit a po\-ly\-nom\-ial-size problem kernel with respect to
  the combined parameter  ``number~$\cols$ of colors'' and ``maximum clique size~$\pw$''.

\nopagebreak
  In particular, there are no polynomial-size problem kernels for the combined parameters $(\pw,k)$ or $(\pw,Q)$, where $Q$ is the number of ``live colors'' (\autoref{def-Q}).
\end{theorem}

\noindent The second part of the theorem follows from the first part since both $k$ and $Q$ are at most $\cols$.

\begin{figure}[t]
  \begin{center}
  \begin{tikzpicture}[x=0.95cm,y=0.95cm]
    \pgfmathtruncatemacro{\loginstances}{3};
    \pgfmathtruncatemacro{\instances}{2^\loginstances};
    \pgfmathsetmacro{\radius}{8/\instances/2}
    \pgfmathsetmacro{\diameter}{2*\radius};
    \pgfmathsetmacro{\height}{0.3};

    \foreach \i in {1,2,3,\instances-1,\instances}{
        \pgfmathsetmacro{\xpos}{(\i-1)*\diameter+\radius};
        \draw[black] (\xpos,0) circle (\radius-0.03);
    }
    \draw (1-\radius,0) node {$x_0$};
    \draw (2-\radius,0) node {$x_1$};
    \draw (3-\radius,0) node {$x_2$};
    \draw (5-\radius,0) node {\huge $\ldots$};
    \draw (\instances-\radius-1,0) node {$x_{t-2}$};
    \draw (\instances-\radius,0) node {$x_{t-1}$};

    \foreach \y in {1,...,\loginstances}{
      \pgfmathtruncatemacro{\maxindex}{2^\y-1}
      \foreach \x in {0,...,\maxindex}{
        \pgfmathsetmacro{\intlength}{2^(\loginstances-\y) * \diameter}
        \pgfmathsetmacro{\ypos}{\radius+\y*\height}
        \draw[interval,lightgray] (\x*\intlength+0.1,\ypos) -- (\x*\intlength+\intlength-0.1,\ypos);
      }
    }
    \draw [decorate,decoration={brace}] (0,\radius+\height) -- (0,\radius+\loginstances*\height) node [xshift=-12,black,midway] {$\log{t}$};

    \pgfmathtruncatemacro{\y}{1}
    \pgfmathtruncatemacro{\x}{1}
    \pgfmathsetmacro{\intlength}{2^(\loginstances-\y) * \diameter}
    \pgfmathsetmacro{\ypos}{\radius+\y*\height}
    \draw[interval,black] (\x*\intlength+0.1,\ypos) -- (\x*\intlength+\intlength-0.1,\ypos);

    \pgfmathtruncatemacro{\y}{2}
    \pgfmathtruncatemacro{\x}{0}
    \pgfmathsetmacro{\intlength}{2^(\loginstances-\y) * \diameter}
    \pgfmathsetmacro{\ypos}{\radius+\y*\height}
    \draw[interval,black] (\x*\intlength+0.1,\ypos) -- (\x*\intlength+\intlength-0.1,\ypos);

    \pgfmathtruncatemacro{\y}{3}
    \pgfmathtruncatemacro{\x}{3}
    \pgfmathsetmacro{\intlength}{2^(\loginstances-\y) * \diameter}
    \pgfmathsetmacro{\ypos}{\radius+\y*\height}
    \draw[interval,black] (\x*\intlength+0.1,\ypos) -- (\x*\intlength+\intlength-0.1,\ypos);
  \end{tikzpicture}
  \end{center}

  \caption{Schematic view of the cross composition for \JIntSel{}. Circles at the bottom represent the~$t$ input instances. Bars at the top represent the auxiliary intervals spanning over the input instances. Here, each of the~$\log{t}$ rows stands for a new color. A solution (black intervals) for the instance must select one interval in each row, thereby selecting one of the~$t$ input instances ($x_2$ in this example).}
  \label{fig:no-kernel}
\end{figure}
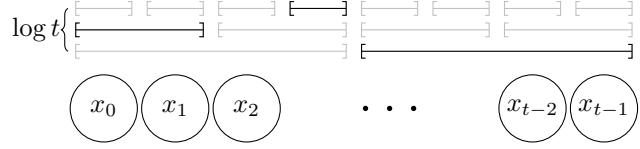

\begin{proof}
  We present a cross composition from the NP-hard starting problem \JIntSel{} with the further restriction that the sought solution size~$k$ equals the number of colors~$\gamma$. We saw in \autoref{cor:NPhwithKColors} in \autoref{sec:dichotomy} that this restriction remains \NP-hard. The framework of \citet{BJK14} allows us to force all of the $t$~input instances~$x_i$ to have the same value for~$k$ and, thus, each instance uses the same color set~$[k]$. 
We assume, without loss of generality, that~$t$~is a power of two (otherwise, we add some ``no''-instances to the list of input instances).
The steps of the cross composition are as follows (see \autoref{fig:no-kernel}):
\begin{noindlist}
  \item Place the start and end points of the $n$~intervals of each input instance~$x_i$ into the integer range~$[i\cdot n, (i+1)\cdot n-1]$.
  \item Introduce~$\log{t}$ extra colors~$k+1,k+2,\dots,k+\log{t}$; the resulting instance then asks for an independent set of size~$k+\log{t}$.
  \item\label{stp:constr-v} For each~$1\leq i\leq\log{t}$, introduce~$2^i$ auxiliary intervals~$v_0,v_1,\dots,v_{2^i-1}$ with color~$k+i$ such that the auxiliary interval~$v_j$ spans exactly over the instances~$x_\ell$ with
    \begin{align*}
      &&j\cdot \frac{t}{2^i}\leq \ell\leq (j+1)\cdot\frac{t}{2^i}-1.
    \end{align*}
\end{noindlist}

To show that this construction is indeed a cross composition for the parameters
``number~$\cols$ of colors'' and ``maximum clique size~$\pw$'', 
observe that $\cols,\pw\leq\max_i{|x_i|}+\log{t}$ and it remains to prove that the constructed instance~$(G,k+\log t)$ is a ``yes''-instance if and only if one of the input instances is a ``yes''-instance.

 First, if the constructed graph~$G$ has a colorful independent set~$I$ of size~$k+\log{t}$, then~$I$ contains an interval of each color. In particular, $I$ contains an auxiliary interval of each of the colors~$k+1$ to $k+\log{t}$. We show that all of the~$k$ non-auxiliary intervals of~$I$ are from the same input instance. To this end, note that, for each~$1\leq i\leq\log{t}$, each auxiliary interval~$v_j$ of color~${k+i}$ spans over exactly
 \begin{align*}
&&   (j+1)\cdot \frac{t}{2^i} - j\cdot\frac{t}{2^i} = \frac{t}{2^i} \text{\quad instances.}
 \end{align*}
Since instances spanned by auxiliary intervals of~$I$ are disjoint, the~$\log{t}$ auxiliary intervals in~$I$ span exactly
\begin{align*}
&&  \sum_{i=1}^{\log t}\frac{t}{2^i}=t-1 \text{\quad instances.}
\end{align*}
Hence, \emph{exactly one} input instance is \emph{not} spanned, implying that all non-auxiliary intervals in~$I$ are from this very instance.

  Second, let~$x_\ell$ be a ``yes''-instance, that is, there is an independent set of size~$k$ in~$x_\ell$ that contains the colors~$[k]$. We extend this to a colorful independent set of size~$k+\log t$ for~$G$. To this end, it is sufficient to add the intervals from a \emph{$(\log t)$-separating} colorful independent set, where a colorful independent set~$I$ is \emph{$i$-separating} for some integer~$i$ if
  \begin{itemize}
  \item no interval in~$I$ spans~$x_\ell$,
  \item $I$ has size~$i$ and contains all colors~$\{k+1,\dots,k+i\}$, and
  \item there is a single interval of color~$k+i$ that is not in~$I$ and covers all instances not spanned by the intervals in~$I$.
  \end{itemize}
  Obviously, there is a $1$-separating colorful independent set, since the intervals with color~$k+1$ separate the input instances in exactly two halves. To complete the proof, it remains to extend this $1$-separating independent set to be $(\log t)$-separating. To this end, we use induction.

  Assume that $I$ is an $i$-separating colorful independent set for some $1\leq i<\log t$. We show how to extend it to be $(i+1)$-separating. The auxiliary intervals in~$I$ span exactly
  \begin{align*}
    &&\sum_{j=1}^{i}\frac{t}{2^j}=t-\frac{t}{2^{i}} \text{\quad instances.}
  \end{align*}
That is, $t/2^{i}$ instances are not spanned by~$I$ but by a single interval of color~$k+i$.
 Since each interval with color~${k+i+1}$ spans~$t/2^{i+1}$~instances and is contained in an interval of color~$k+i$, there are precisely two intervals of color~${k+i+1}$ that span the instances not spanned by~$I$. Since they are disjoint, one of them does not span~$x_\ell$, add this interval to~$I$.
\end{proof}

\paragraph{Polynomial-Size Problem Kernel on Proper Interval Graphs.}
\looseness=-1 We restrict \JIntSel{} to proper interval graphs, which remains \NP-hard, as shown in \autoref{sec:dichotomy}. Here, the negative result of \autoref{thm:uig-nopoly} collapses: the following polynomial-time  data reduction routine produces a problem kernel containing $4k^2\cdot\pw$~intervals, where $\pw$ is the maximum clique size in the input graph.

\begin{comment}
\end{comment}

\begin{rrule}\label{rr:color packing}
  For a graph~$G$ and a color~$c$, let $G[c]$~denote the subgraph of~$G$ induced by the intervals of color~$c$.

  If $G[c]$  has an independent set of size at least~${2k-1}$, then remove all intervals of color~$c$ and decrease~$k$ by one.
\end{rrule}

\noindent In order to show that a reduction rule is \emph{correct}, one
has to show that the output instance is a ``yes''-instance if and only
if the input instance is.

\begin{lemma}\label{lem:color packing}
  \autoref{rr:color packing} is correct and can be applied exhaustively in $O(n)$~time.
\end{lemma}
\begin{proof}
  Let $(G',k-1)$ denote the instance produced by \autoref{rr:color packing} from an instance~$(G,k)$ by removing all intervals of a color~$c$ from~$G$. Clearly, a colorful independent set~$I$ for~$G$ is also a colorful independent set for~$G'$ if we remove the interval with color~$c$ from~$I$. Hence, if~$(G,k)$ is a ``yes''-instance, then so is~$(G',k-1)$.

  In the following, let~$(G',k-1)$ be a ``yes''-instance with solution~$I$ and let~$I_c$ denote an independent set of size~$2k-1$ in~$G[c]$. If $|I|\geq k$, then $(G,k)$ is a ``yes''-instance. Otherwise, $|I|\leq k-1$.
  Furthermore, $I$ does not contain an interval with color~$c$.
  Consider an interval~$u\in I$. If $N_G(u)\cap I_c$ contains at least three intervals~$x,y,z$, then~$G[\{u,x,y,z\}]$ is a~$K_{1,3}$, contradicting~$G$ being a proper interval graph. Therefore, each interval in~$I$ overlaps at most two intervals in~$I_c$. Hence, the intervals in~$I$ overlap at most~$2|I|\leq 2(k-1)<|I_c|$ intervals, implying that~$I_c\setminus N[I]\ne\emptyset$. Thus, there is an interval in~$I_c$ that can be added to~$I$, thereby obtaining a solution for~$(G,k)$.

  It remains to argue the claimed running time. To this end, note that a maximum independent set in~$G[c]$ can be computed in~$O(n_c)$~time with~$n_c:=|V(G[c])|$, since $G[c]$ is an ordinary (that is, monochromatic) proper interval graph. Hence, computing maximum independent sets for all colors can be done in $O(n)$~time in total. Since applying the rule for one color does not affect other colors, this application is exhaustive, that is, \autoref{rr:color packing} is not applicable to the resulting instance.
\end{proof}

\noindent In order to prove the problem kernel bound, we further need the following trivial ``data reduction rule'' that returns a ``yes''-instance if we can greedily find an optimal solution.
\begin{rrule}\label{rr:maximal CIS}
  Let~$I$ be a maximal colorful independent set of~$G$.
  If~$|I|\geq k$, then return a small trivial ``yes''-instance.
\end{rrule}

\begin{lemma}\label{lem:maximal CIS correct}
  \autoref{rr:maximal CIS} is correct can be applied in $O(n)$~time.
\end{lemma}

\begin{proof}
  The correctness of \autoref{rr:maximal CIS} is obvious. It remains to
  prove the running time.

A maximal colorful independent set of~$G$ can be found by greedily picking the first-ending valid interval~$v$ into an independent set~$I$ and deleting all intervals that overlap~$v$. Herein, we can keep a size-$\gamma$ array whose $i$-th entry is~$1$ if color~$i$ is already used. Using this array, we can check in constant time whether an interval is valid for inclusion in~$I$. Moreover, since invalid vertices do not become valid again, the whole procedure can be executed in $O(n)$~time, given that the intervals are sorted.
\end{proof}

\noindent Given these two data reduction rules,  we can now prove the following theorem.

\begin{theorem}\label{thm:splitUIG kernel} %
  \JIntSel{} on proper interval graphs admits a problem kernel with at most $4k^2\cdot\pw$~intervals that is computable in $O(n)$~time. Herein, $\pw$ is the maximum clique size of the input graph.
\end{theorem}
\begin{proof}
  To show the problem kernel bound, consider an instance~$(G,k)$ of
  \JIntSel{} that is reduced with respect to \autoref{rr:color packing}
  and to which \autoref{rr:maximal CIS} has been applied. It follows
  that there is a maximal colorful independent set~$I$ of~$G$
  with~$|I|<k$.  Since~$G$ is a \emph{proper} interval graph, the neighborhood of each vertex~$v$ can be partitioned
  into two cliques: one consisting of intervals containing~$v_s$, one
  consisting of intervals containing~$v_e$.  Thus, each vertex in~$I$
  has at most~$2\pw-1$ neighbors and, hence, we can bound~$|N[I]|\leq
  2k\pw$.

  Now, let~$X:=V(G)\setminus N[I]$ and let~$G':=G[X]$. Then, since~$I$ is maximal, all intervals in~$X$ have a color that appears in~$I$, of which there are at most~$k-1$. For each color~$c$ of these, let~$G'[c]$ denote the subgraph of~$G'$ that is induced by all intervals of color~$c$ in~$X$ and let~$I_c$ denote a maximum independent set of~$G'[c]$. Since~$G$ is reduced with respect to \autoref{rr:color packing}, $|I_c|\leq 2(k-1)$. Again, since~$G'[c]$ is a proper interval graph, each interval~$u\in I_c$ has at most~${2\pw-1}$ neighbors in~$G'[c]$. Thus, the total number of intervals in~$G'[c]$ is at most~$4(k-1)(\pw-1)$. Since~$G'$ contains at most~$k-1$ colors, we can bound~$|V(G')|\leq4(k-1)^2(\pw-1)$, implying a bound of~$|V(G)|+|N[I]|\leq 4(k-1)^2(\pw-1)+2k\pw\leq 4k^2\pw$ for the number of intervals in~$G$.

  The running time bound follows from \autoref{lem:color packing} and \autoref{lem:maximal CIS correct}.
\end{proof}

\section{2-Union Independent Set}\label{sec:mcis}
\label{sec:independentsetin2uniographs}
In \autoref{sec:jisp}, we studied the \JIntSel{} problem, which is equivalent to \MCIS{} where one of the two input interval graphs is a cluster graph. In this section, we investigate the parameterized complexity \MCIS{}. 

\autoref{cor:mcis-veryhard} has already shown that \MCIS{} is NP-hard even if the maximum clique size of both input interval graphs and the maximum vertex degree are at most two. Moreover, we already know that \MCIS{} is W[1]-hard with respect to the %
parameter~$k$ \citep{Jia10}. Hence, with respect to these three parameters, \MCIS{} is unlikely to be \fp{} tractable.

In contrast, this section shows how the compactness of the input
interval graphs affects the computational complexity of \MCIS{}. To this
end, as before, let $\CompMax$~be the minimum number such that both
input interval graphs are $\CompMax$-compact and let $\CompMin$~be the
minimum number such that at least one of both input interval graphs is
$\CompMin$-compact (see \autoref{def:compactness}).

\looseness=-1 First, in \autoref{fpt-mcis}, we show a \fp{} algorithm with respect to the parameter~$\CompMin$. The algorithm is an adaption of our algorithm for \MColIS{} (\autoref{thm:mcolis-fpt}) to \MCIS{}. In the analysis of its complexity, the parameter~$\CompMin$ naturally arises as complexity measure.

Second, in  \autoref{mcis-dr}, we show a simple polynomial-time data reduction rule for \MCIS{}. Again, in the analysis of its effectiveness, the parameter~$\CompMax$ naturally arises as complexity measure.

Since in both applications, compactness-related parameters arose quite
naturally, we suspect that the parameter may be useful in the
development in \fp{} algorithms for other NP-hard problems on interval
graphs.

\subsection{A Dynamic Program for 2-Union Independent Set}\label{fpt-mcis}
We describe an algorithm that solves \MCIS{} in $O(2^{\CompMin} \cdot n)$~time. To this end, we reformulate \MCIS{} as a special case of \MColIS{} and then solve the resulting instance using the dynamic program~\eqref{DP-Qp} from \autoref{sec:DP} (\autoref{thm:mcolis-fpt}).

An instance of \MCIS{} can be solved by an algorithm for \MColIS{} as follows: without loss of generality, assume that of the input interval graphs $G_2$~is $\CompMin$-compact. We interpret each number in $[\CompMin]$ as a color and give the input graph~$G_1$ as input to \MColIS{} such that each vertex~$v$ of~$G_1$ gets the colors corresponding to the numbers contained in the interval that represents~$v$ in~$G_2$. Then a solution for \MColIS{} is a solution for \MCIS{} and vice versa: 

\begin{itemize}
\item Two vertices~$v$ and~$w$ may be together in a solution of \MCIS{} if and only if their intervals neither intersect in~$G_1$ nor in~$G_2$.
\item Two vertices~$v$ and~$w$ may be together in a solution of \MColIS{} if and only if neither their intervals in~$G_1$ intersect nor their color lists  intersect (which are precisely their intervals in~$G_2$).
\end{itemize}
We stated earlier that \MColIS{} is a more general problem than \MCIS{}. This now becomes clear: whereas \MColIS{} allows arbitrary color lists, the instances generated from \MCIS{} only use intervals of natural numbers as color lists.

To execute the transformation from \MCIS{} to \MColIS{}, we just take each interval of~$G_2$ and add the numbers that it contains to the color list of the corresponding vertex in~$G_1$. Since each interval in~$G_2$ contains at most $\CompMin$~numbers, the transformation from \MCIS{} to \MColIS{} is executable in $O(\CompMin\cdot n)$~time. The resulting \MColIS{} instance has $\CompMin$~colors and, by \autoref{thm:mcolis-fpt}, is solvable in additionally $O(2^{\CompMin}\cdot n)$~time.

\begin{theorem}\label{thm:fpt-2union}
  \MCIS{} is solvable in $O(2^{\CompMin}\cdot n)$ time when at least one input interval graph is $\CompMin$-compact.
\end{theorem}

\subsection{Polynomial-Time Preprocessing}\label{mcis-dr}
We provide polynomial-time data reduction for \MCIS{}. It will turn out that the presented data reduction rule yields a polynomial-size problem kernel for the parameter~$\CompMax$, where both input interval graphs~$G_1$, $G_2$ are $\CompMax$-compact.

The intuition behind the data reduction rule is simple: assume that we have a vertex that is represented by the interval~$v$ in the first input interval graph~$G_1$ and by~$v'$ in the second input interval graph~$G_2$. Moreover, assume that there is another vertex represented by the intervals~$u$ in~$G_1$ and~$u'$ in~$G_2$. Then, if $v\subseteq u$ and $v'\subseteq u'$, we would never choose the vertex represented by $u$~and~$u'$ into a maximum independent set, as it ``blocks'' a superset of vertices for inclusion into a maximum independent set compared to the vertex represented by~$v$ and~$v'$. Hence, we delete the intervals $u$ and~$u'$.

To lead this intuitive idea to a problem kernel, we introduce the
concept of the \emph{signature} of a vertex, give a reduction rule that
bounds the number of vertices having a given signature, and finally
bound the number of signatures in a 2-union graph.

\begin{definition}\label{def:signature}
  Let~$(G_1,G_2,k)$ denote an instance of \MCIS{} and let $v$ be a vertex of $G_1$ and $G_2$.
  The \emph{signature}~$\sig(v)$ of~$v$ is a four-dimensional vector $(-v_s,v_e,\allowbreak -v_s',v_e')$, where $v_s$ and $v_e$ are $v$'s start and end points in~$G_1$, and $v_s'$ and $v_e'$ are its start and end points in~$G_2$.
\end{definition}

\begin{rrule}\label{rule:del-supersig}
  Let $(G_1,G_2,k)$ denote an instance of \MCIS{}. For each pair of vertices~$u,v$ of~$G_1$ and~$G_2$ such that~$\sig(v)\leq \sig(u)$ (component-wise), delete~$u$ from~$G_1$ and~$G_2$.
\end{rrule}
\begin{lemma}\label{lem:del-supersig}
  \autoref{rule:del-supersig} is correct and can be applied in $O(n\log^2{n})$~time.
\end{lemma}
\begin{proof}
  Let~$(G_1,G_2,k)$ be an instance of \MCIS{} and let~$u,v$ be vertices
  of~$G_1$ and~$G_2$ such that~$\sig(v)\leq\sig(u)$.
  Observe that this implies $v_s\geq u_s$, $v_e\leq u_s$, $v_s'\geq u_s'$, and $v_e'\leq u_e'$. Hence, $N_{G_1}[v]\subseteq N_{G_1}[u]$ and~$N_{G_2}[v]\subseteq N_{G_2}[u]$ and, therefore, $N_{G_1\cup G_2}[v]\subseteq N_{G_1\cup G_2}[u]$. Hence, instead of choosing~$u$ into an independent set, we can always choose~$v$. Therefore, it is safe to delete~$u$.

  Regarding the running time,
  \citet{journal-KLP75} showed that the set of maxima of $n$~vectors in $d$~dimensions can be computed using $O(n\log^{d-2}{n})$~comparisons, directly implying the stated running time.
\end{proof}

\begin{theorem}\label{thm:sigkern} %
  Let $(G_1,G_2,k)$ be an instance of \MCIS{} such that~$G_1$ and~$G_2$
  are $\CompMax$-compact. 

  Then, a problem kernel with
  $\CompMax^3$~vertices can be constructed in~$O(n\log^2 n)$~time.
  The size reduces to $2\CompMax^2$~vertices if one of the input graphs
  is proper interval.
\end{theorem}

\begin{proof}
  Let~$(G_1,G_2,k)$ be an instance of \MCIS{}. We assume that~$G_1$
  and~$G_2$ have been preprocessed according to
  \autoref{obs:shrinkints}, that is, at each position of the interval
  representations of~$G_1$ and~$G_2$, there is an interval start point
  as well as an interval end point.  The problem kernel
  $(G^*_1,G^*_2,k)$~is then obtained from~$(G_1,G_2,k)$ by applying
  \autoref{rule:del-supersig} to~$(G_1,G_2,k)$.  By definition
  of~$G_1^*$ and~$G_2^*$, the graphs~$G_1^*$ and~$G_2^*$ contain at
  most one vertex of each signature. Hence, it is sufficient to show
  that there are at most $\CompMax^3$ different signatures
  corresponding to vertices in the new instance~$(G_1^*,G_2^*,k)$.

  Consider the set~$\mathcal{S}_{i,j}$ of all
  signatures~$s=(-v_s,v_e,-v_s',v_e')$ with $v_s=i$ and $v_s'=j$ such
  that $v$~remains in~$G_1^*$ and~$G_2^*$. If
  $|\mathcal{S}_{i,j}|>\CompMax$, then we find~$s_1,s_2\in \mathcal
  S_{i,j}$ such that $s_1$ and~$s_2$ agree in the second or fourth
  coordinate, since there are at most $\CompMax$~possible values for
  each of them. Since then $s_1$ and $s_2$ agree in three coordinates,
  it follows that either $s_1\leq s_2$ or $s_2\leq s_1$, contradicting
  the assumption that \autoref{rule:del-supersig} has been applied
  to~$(G_1^*,G_2^*,k)$.  Obviously, there are at most~$\CompMax^2$ sets
  of signatures~$\mathcal{S}_{i,j}$ and, thus, there are at most
  $\CompMax^3$ signatures in total.

  \medskip\noindent In the following, we show that the described
  instance has at most $2\CompMax^2$~vertices if $G_1$ is a proper
  interval graph.  To this end, we will use two relations~$R_1,R_2$
  between pairs~$(i,j)\in[\CompMax]\times[\CompMax]$ and signatures
  corresponding to vertices of~$G^*_1$.  We then show that each signature is the image under one of~$R_1$ and~$R_2$ and that both
  relations are in fact functions, that is, they map each pair to at
  most one signature.
  This proves that there are at most~$2\cdot
  |[\CompMax]\times[\CompMax]|=2\CompMax^2$~signatures corresponding
  to vertices in~$G_1^*$. Since~$G_1^*$~contains at most one vertex
  per signature, the theorem will follow.

  \looseness=-1 We define the relations $R_1$ and $R_2$ as follows: for a
  pair~$(i,j)$, the relation~$R_1$ associates $(i,j)$ with all
  signatures~$s=\sig(v)=\vsig{v}$ of~$\mathcal{S}_{i,j}$ that
  minimize~$v'_e$.  For all signatures~$s=\sig(w)=\vsig{w}$ with
  $w$~remaining in~$G_1^*$ and that are \emph{not} images under~$R_1$,
  the relation~$R_2$ associates~$(w_e,w'_s)$ with~$s$. By definition,
  every signature is the image of some pair under \emph{either}~$R_1$
  or~$R_2$.

  Observe that $R_1$~is a function: since~$G^*_1$ and~$G^*_2$ are
  reduced with respect to \autoref{rule:del-supersig}, for each
  pair~$(i,j)\in[\CompMax]\times[\CompMax]$, there is at most one
  signature~$s=\vsig{v}\in\mathcal{S}_{i,j}$ that minimizes $v_e'$.

  It remains to show that $R_2$~is also a function.  Towards a
  contradiction, assume that $R_2$ maps some pair to two
  signatures. Then, there are distinct vertices~$v$ and~$w$ in~$G_1^*$
  with signatures~$s_1:=\sig(v)=\vsig{v}$ and $s_2:=\sig(w)=\vsig{w}$
  such that $(v_e,v'_s)=(w_e,w'_s)$.  Since the vertices $v$ and~$w$
  are in~$G^*_1$, we know that $v_s\neq w_s$, since otherwise $s_1\leq
  s_2$ or $s_2\leq s_1$.  By symmetry, let $w_s<v_s$. Since
  $s_2\in\mathcal S_{w_s,w_s'}$, and, therefore, $\mathcal
  S_{w_s,w_s'}$~is nonempty, there is a
  signature~$s_3:=R_1(w_s,w_s')=\sig(u)=\vsig{u}$. Since $s_2$~is an
  image under~$R_2$, it is not an image under~$R_1$ and, thus, we have
  $s_2\ne s_3$.  %
  The definition of~$R_1$ implies
  $u_s=w_s%
  $, $u_s'=w_s'$, and~$u'_e\leq w'_e$. Therefore, we have $u_e>w_e%
  $
  since, otherwise, $s_3\leq s_2$.

  \begin{figure}
    \centering
    \begin{tikzpicture}
      \draw[{-]}] (-2,0.3)--node[midway,above]{$x$}(-1,0.3);
      \draw[{[-}] (1,0.3)--node[midway,above]{$y$}(2,0.3);
      \draw[interval] (-0.75,0.3)--node[midway,above]{$v$}(0.75,0.3);
      \draw[interval] (-1,0)--node[midway,below]{$u$}(1,0);
    \end{tikzpicture}
    \caption{The constellation $u_s= x_e< v_s$ and $v_e<y_s=u_e$ that
      induces a~$K_{1,3}$ in~$G_1$.}
    \label{fig:interclaw}
  \end{figure}
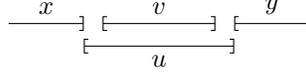

  \looseness=-1 Since $u_s=w_s<v_s$ and $v_e=w_e<u_e$, we now have a constellation~$u_s<v_s\leq v_e<u_e$ of
  intervals that contradicts $G_1$ being a proper interval graph:
  since $G_1$~has been preprocessed according to
  \autoref{obs:shrinkints}, the start point~$u_s$ is also the end
  point~$x_e$ of some interval~$x$ and the end point~$u_e$ is also the
  start point~$y_s$ of some interval~$y$.  Note that $x_e=u_s<v_s\leq v_e<u_e=y_s$ implies that $x,y,v$, and $u$ are pairwise distinct.
  This, as depicted in
  \autoref{fig:interclaw}, implies that $G_1$~contains a~$K_{1,3}$ as
  induced subgraph, which contradicts~$G_1$ being a proper interval
  graph. The $K_{1,3}$~consists of the
  central vertex~$u$ and the leaves~$v,x,y$. 
\end{proof}

\noindent We can generalize \autoref{thm:sigkern} for the problem of finding an independent set of \emph{weight} at least~$k$: we only have to keep that vertex for each signature in the graph that has the highest weight.
Since there are at most~$\CompMax^4$ different signatures, we obtain a problem kernel with $\CompMax^4$~vertices for the weighted variant of \MCIS{}.

\begin{comment}
\end{comment}

\section{Experimental Evaluation}
\label{sec:experiments}\label{sec:exp}

In this section, we aim for giving a proof of concept by demonstrating to which extent instances of \MColIS{} are solvable within an acceptable time frame of five minutes. Herein, we chose \MColIS{} (see \autoref{sec:mcolisdef} for the definition) since it is the most general problem studied in our work and algorithms for it also solve \MCIS{} and \JIntSel{}. 

\looseness=-1 We implemented the dynamic programming algorithm~\eqref{DP-Qp} from
\autoref{sec:DP} that solves \MColIS{} in $O(2^Q\cdot n)$~time and
$O(2^Q\ell+\gamma c)$~space (\autoref{thm:mcolis-fpt}), where $\gamma$
is the number of colors, $c$~is the compactness of the input interval
graph, $\ell$~is the maximum length of an interval, and $Q$~is the
structural parameter ``maximum number of live colors''
(\autoref{def-Q}). We applied the implemented algorithm to randomly
generated instances.

Note that we abstained from implementing our data reduction rules (Section \ref{sec:kernelization} and \ref{mcis-dr}), since they do not apply to the most general form of \MColIS{}, which we aim to experiment with.

\paragraph{Implementation Details.} 
 The implementation of the algorithm is based on recurrence~\eqref{DP-Qp} from \autoref{sec:DP}, but allows the vertices to have weights and finds a colorful independent set of maximum weight. The source code uses about 700 lines of C++ and is freely available.\footnote{\url{http://fpt.akt.tu-berlin.de/cis/}} 
The experiments were run on a computer with a 3.6\,GHz Intel Xeon processor and 64\,GiB RAM under Linux~3.2.0, where the source code has been compiled using the GNU C++ compiler in version~4.7.2 and using the highest optimization level~(-O3). 

\paragraph{Data.} \looseness=-1 In order to test the influence of various parameters on the running time and memory usage of the algorithm, we evaluated the algorithm on artificial, randomly generated data. To generate random interval graphs, we use a model that is strongly inspired by \citet{Sch88}. However, while \citet{Sch88} chooses integer interval endpoints uniformly at random without repetitions from~$[2n]$, we choose integer interval endpoints uniformly at random from~$[c]$, where~$c$~is a maximum compactness chosen in advance. It then remains to assign colors and weights to the vertices.

In detail, to generate a random interval graph, we fix a maximum compactness~$c$, a maximum number~$\gamma$ of colors, and a number~$n$ of intervals to generate. We then randomly generate $n$~intervals: for each interval~$v$, we choose a start point~$v_s$ and an end point~$v_e$ uniformly at random from $[c]$. Then, we add each color in $[\gamma]$ to the color list of~$v$ with probability~$1/2$ and uniformly at random assign $v$~a weight from~1 to~10.

To interpret the experimental results, it is important to make some structural observations about the data generated by this random process.

\begin{noindlist}

\item \looseness=-1 The maximum interval length~$\ell$ is at most $c-1$. Moreover, with a growing number~$n$ of generated intervals, the probability $(1-1/c^2)^n$ of \emph{not} generating an interval that indeed has length~$c-1$ approaches zero.
  That is, we expect the chosen parameter~$c\approx\ell+1$ to roughly linearly influence the memory usage of the algorithm (\autoref{thm:mcolis-fpt}).

\item The maximum number of live colors~$Q$ is at most the number~$\gamma$ of colors. However, since every interval contains each color with equal probability, with increasing number~$n$ of intervals we will have $Q\approx\gamma$. Hence, we expect $\gamma$ to exponentially influence the running time and memory usage of the algorithm (\autoref{thm:mcolis-fpt}). 
\item The sizes of the generated vertex color lists follow a binomial distribution. The expected color list size is~$\gamma/2$.
\end{noindlist}

\paragraph{Experimental Results.}

We generated three data sets by varying each time one of the parameters~$\{n,\gamma,c\}$ and keeping the other two constant. We applied our algorithm to find a maximum colorful independent set in each of the graphs.

\begin{figure*}
  \begin{tikzpicture}
    \begin{axis}[change y base, change x base, y unit=s, ymax=999,
      xmax=20, xlabel=number~$\gamma$ of colors, ylabel=running time,
      width=0.45\textwidth, ymax=500, xmax=18.5]

      \addplot[color=black,mark=+,only marks] table [x=COLS, y=T1]{INTS	LEN	COLS	MEM	CLIQUES	T1	T2	T3	T4	T5
100000 1000 10	7.73841	990	1.78	1.77	1.76	1.78	1.78	
100000 1000 11	15.465	989	3.55	3.55	3.53	3.52	3.54	
100000 1000 12	30.9337	989	7.05	7.05	7.06	7.05	7.06	
100000 1000 13	61.8087	991	14.09	14.1	14.08	14.08	14.08	
100000 1000 14	123.871	991	28.05	28.05	28.05	28.05	28.06	
100000 1000 15	247.746	991	55.92	55.94	55.95	55.93	55.99	
100000 1000 16	494.996	990	111.57	111.91	111.67	111.6	111.59	
100000 1000 17	988.996	988	223.17	223.14	223.13	242.03	223.23	
100000 1000 18	1990	994	445.72	445.89	445.89	445.92	445.77	
100000 1000 19	3948	987	891.69	891.77	891.76	891.87	891.64	
100000 1000 20	7928	991	1793.43	1793.84	1793.56	1793.41	1793.65	
};
  \end{axis}
\end{tikzpicture}\hfill{}
\begin{tikzpicture}
  \begin{axis}[change y base, change x base, use units, y unit=GB,
    axis base prefix={axis y base -3 prefix {}},
    xlabel=number~$\gamma$ of colors, ylabel=memory usage,
    width=0.45\textwidth, xmax=18.5]
    \addplot[color=black,mark=+,only marks] table [x=COLS, y=MEM]
    {INTS	LEN	COLS	MEM	CLIQUES	T1	T2	T3	T4	T5
100000 1000 10	7.73841	990	1.78	1.77	1.76	1.78	1.78	
100000 1000 11	15.465	989	3.55	3.55	3.53	3.52	3.54	
100000 1000 12	30.9337	989	7.05	7.05	7.06	7.05	7.06	
100000 1000 13	61.8087	991	14.09	14.1	14.08	14.08	14.08	
100000 1000 14	123.871	991	28.05	28.05	28.05	28.05	28.06	
100000 1000 15	247.746	991	55.92	55.94	55.95	55.93	55.99	
100000 1000 16	494.996	990	111.57	111.91	111.67	111.6	111.59	
100000 1000 17	988.996	988	223.17	223.14	223.13	242.03	223.23	
100000 1000 18	1990	994	445.72	445.89	445.89	445.92	445.77	
100000 1000 19	3948	987	891.69	891.77	891.76	891.87	891.64	
100000 1000 20	7928	991	1793.43	1793.84	1793.56	1793.41	1793.65	
};
  \end{axis}
\end{tikzpicture}
\caption{Dependence of running time and space requirements of the
  dynamic program~(\ref{DP-Qp}, \autoref{sec:DP}) for \MColIS{} on the
  number~$\gamma$ of colors in the input interval graph, each having
  $10^5$ intervals and being $10^3$-compact.}
\label{fig:vary-gamma}
\end{figure*}
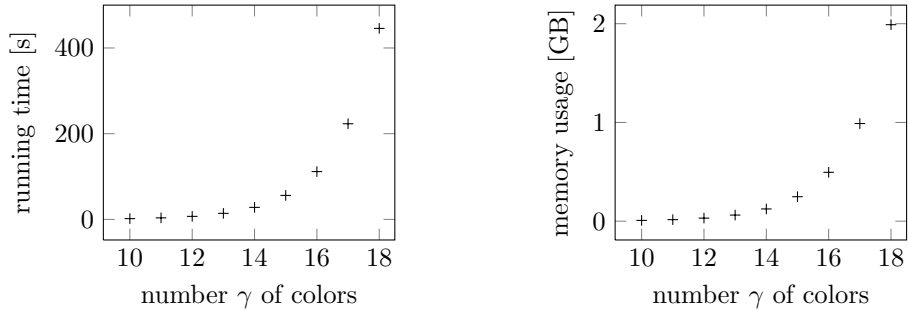

For the first data set, we let the number~$\gamma$ of colors vary between $10$ and $18$ and fixed $n=10^5$ and $c=10^3$. \autoref{fig:vary-gamma} clearly exhibits the exponential dependence of running time and memory usage on the number~$\gamma$ of colors, which both roughly double when increasing the number of colors by one. We see that, in this setup, we can solve \MColIS{} within a time frame of five minutes for $\gamma\leq17$.

\begin{figure*}
  \begin{tikzpicture}
    \begin{axis}[change y base, change x base, axis base prefix={axis x base -3 prefix {}}, 
      y unit=s, x unit=10^3, xlabel=number~$n$ of intervals,
      ylabel=running time, width=0.45\textwidth, ymin=0]
      
      \addplot[color=black,mark=+,only marks, each nth point=2] table
      [x=N, y=T, col sep=comma] {N,T,Mem
105000,58.79,246.996
110000,61.54,248.996
115000,64.31,248.496
120000,67.1,248.746
125000,69.86,248.246
130000,72.66,248.996
135000,75.43,247.996
140000,78.33,248.746
145000,81.09,248.246
150000,83.85,247.746
155000,86.7,248.746
160000,89.42,247.996
165000,92.18,249.246
170000,94.96,247.996
175000,97.73,248.746
180000,100.54,248.496
185000,103.29,249.246
190000,106.09,247.746
195000,108.88,248.996
200000,111.7,248.746
205000,114.46,249.246
210000,117.28,248.996
215000,120.05,248.996
220000,122.84,248.996
225000,125.57,249.246
230000,128.4,248.996
235000,131.14,249.246
240000,134.02,249.496
245000,136.8,248.996
250000,139.53,249.496
255000,142.28,249.246
260000,155.04,249.246
265000,148.01,249.746
270000,150.65,249.746
275000,153.49,248.746
280000,156.2,249.496
285000,159.12,249.496
290000,161.88,249.746
295000,164.55,249.496
300000,167.32,248.746
305000,170.08,249.496
310000,172.91,249.246
315000,175.81,249.746
320000,178.69,249.496
325000,181.34,248.746
330000,184.14,249.746
335000,186.97,249.746
340000,190.02,248.996
345000,192.48,249.246
350000,195.26,249.496
355000,197.97,248.996
360000,200.84,249.496
365000,203.54,249.496
370000,206.36,249.996
375000,209.14,249.746
380000,211.92,249.746
385000,214.7,249.996
390000,217.53,249.496
395000,220.25,249.746
400000,223.04,249.246
405000,225.84,249.746
410000,228.61,250.246
415000,231.34,249.496
420000,234.18,249.996
425000,237.06,249.246
430000,239.79,249.496
435000,242.52,249.496
440000,245.48,249.746
445000,248.39,249.496
450000,250.99,248.996
455000,253.68,249.496
460000,256.44,249.746
465000,259.24,249.496
470000,262.08,249.246
475000,265.02,249.246
480000,267.53,249.996
485000,270.39,249.996
490000,273.2,249.746
495000,276.02,249.996
500000,278.81,249.996
505000,281.57,248.996
510000,284.29,249.496
515000,287.1,249.246
520000,289.9,249.746
525000,292.67,249.746
530000,295.37,249.746
535000,298.2,249.496
540000,300.98,249.996
545000,303.93,249.246
550000,306.55,249.496
555000,309.33,249.746
560000,312.36,249.746
565000,314.98,249.746
570000,317.67,249.496
575000,320.49,249.746
580000,323.29,249.996
585000,326.16,249.746
590000,328.81,249.746
595000,331.68,249.496
600000,334.41,249.746
};
    \end{axis}
  \end{tikzpicture}\hfill{}
  \begin{tikzpicture}
    \begin{axis}[change y base, change x base, axis base prefix={axis x base -3 prefix {}},%
      y unit=MB, x unit=10^3, ymin=240, ymax=260, xlabel=number~$n$
      of intervals, ylabel=memory usage, width=0.45\textwidth]

      \addplot[color=black,mark=+,only marks, each nth point=2] table
      [x=N, y=Mem, col sep=comma] {N,T,Mem
105000,58.79,246.996
110000,61.54,248.996
115000,64.31,248.496
120000,67.1,248.746
125000,69.86,248.246
130000,72.66,248.996
135000,75.43,247.996
140000,78.33,248.746
145000,81.09,248.246
150000,83.85,247.746
155000,86.7,248.746
160000,89.42,247.996
165000,92.18,249.246
170000,94.96,247.996
175000,97.73,248.746
180000,100.54,248.496
185000,103.29,249.246
190000,106.09,247.746
195000,108.88,248.996
200000,111.7,248.746
205000,114.46,249.246
210000,117.28,248.996
215000,120.05,248.996
220000,122.84,248.996
225000,125.57,249.246
230000,128.4,248.996
235000,131.14,249.246
240000,134.02,249.496
245000,136.8,248.996
250000,139.53,249.496
255000,142.28,249.246
260000,155.04,249.246
265000,148.01,249.746
270000,150.65,249.746
275000,153.49,248.746
280000,156.2,249.496
285000,159.12,249.496
290000,161.88,249.746
295000,164.55,249.496
300000,167.32,248.746
305000,170.08,249.496
310000,172.91,249.246
315000,175.81,249.746
320000,178.69,249.496
325000,181.34,248.746
330000,184.14,249.746
335000,186.97,249.746
340000,190.02,248.996
345000,192.48,249.246
350000,195.26,249.496
355000,197.97,248.996
360000,200.84,249.496
365000,203.54,249.496
370000,206.36,249.996
375000,209.14,249.746
380000,211.92,249.746
385000,214.7,249.996
390000,217.53,249.496
395000,220.25,249.746
400000,223.04,249.246
405000,225.84,249.746
410000,228.61,250.246
415000,231.34,249.496
420000,234.18,249.996
425000,237.06,249.246
430000,239.79,249.496
435000,242.52,249.496
440000,245.48,249.746
445000,248.39,249.496
450000,250.99,248.996
455000,253.68,249.496
460000,256.44,249.746
465000,259.24,249.496
470000,262.08,249.246
475000,265.02,249.246
480000,267.53,249.996
485000,270.39,249.996
490000,273.2,249.746
495000,276.02,249.996
500000,278.81,249.996
505000,281.57,248.996
510000,284.29,249.496
515000,287.1,249.246
520000,289.9,249.746
525000,292.67,249.746
530000,295.37,249.746
535000,298.2,249.496
540000,300.98,249.996
545000,303.93,249.246
550000,306.55,249.496
555000,309.33,249.746
560000,312.36,249.746
565000,314.98,249.746
570000,317.67,249.496
575000,320.49,249.746
580000,323.29,249.996
585000,326.16,249.746
590000,328.81,249.746
595000,331.68,249.496
600000,334.41,249.746
};
    \end{axis}
  \end{tikzpicture}
  \caption{Dependence of running time and space requirements
    of the dynamic program~(\ref{DP-Qp}, \autoref{sec:DP}) for
    \MColIS{} on the number~$n$ of intervals in the input interval
    graph, each being colored with subsets of~$\{1,\dots,15\}$ and being
    $10^3$-compact.}
    \label{fig:vary-n}
\end{figure*}
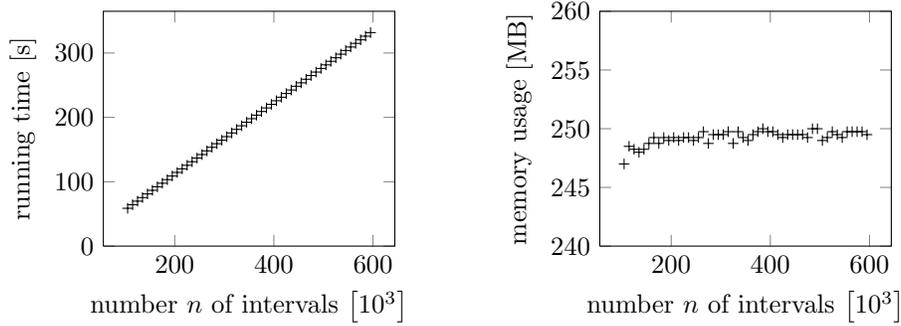

For the second data set, we let the number~$n$ of intervals vary between~$10^5$ and~${6\cdot 10^5}$. We again fixed~$c=10^3$. We chose $\gamma=15$ as number of colors. As expected, \autoref{fig:vary-n} shows a roughly linear dependence of the running time on the number~$n$ of intervals. %
Moreover, the memory usage is almost constantly about 250\,MB with a slight increase, since we left the compactness~$c$ constant and  with increasing number~$n$ of intervals, the maximum interval length~$\ell$ approaches the maximum compactness~$c$.

\begin{figure*}
\begin{tikzpicture}
  \begin{axis}[change y base, change x base, %
    legend style={at={(0.05,0.05)}, anchor=south west}, 
    y unit=s, %
    width=0.45\textwidth,
    ymin=50, ymax=65,
    xlabel=compactness~$c$, ylabel=running time]

    \addplot[color=black,mark=+, only marks, each nth point=2] table
      [x=M, y=T, col sep=comma] {M,Mem,T
100,25.2496,56.11
110,27.7496,55.84
120,30.2495,55.67
130,32.7495,56.65
140,35.2495,56.58
150,37.7494,56.59
160,40.2494,56.47
170,42.7493,56.4
180,45.2493,56.3
190,47.4993,56.18
199,49.9992,56.11
210,52.7492,56.02
220,55.2492,55.93
230,57.7491,55.85
240,60.2491,55.78
249,62.499,55.67
259,64.999,56.64
270,67.749,56.64
280,69.9989,56.68
289,72.2489,56.61
300,75.2489,56.57
308,77.2488,56.55
316,79.2488,56.54
329,82.4987,56.62
339,84.9987,56.43
350,87.4987,56.39
360,89.9986,56.34
369,92.4986,56.32
378,94.7486,56.31
388,97.2485,56.25
399,99.7485,56.17
407,101.998,56.16
418,104.748,56.11
429,107.248,56.09
438,109.748,56.02
449,112.248,55.98
456,114.248,55.96
469,117.248,55.94
480,120.248,55.9
487,121.998,55.82
497,124.498,55.81
509,127.248,55.75
518,129.498,56.66
528,131.998,59.52
537,134.498,56.66
548,137.248,56.69
558,139.498,56.67
568,141.998,56.68
576,144.248,56.65
585,146.498,56.65
596,149.248,56.67
606,151.748,56.61
615,153.998,56.66
625,156.248,56.62
637,159.498,56.59
648,162.248,56.77
656,164.247,56.56
667,166.497,56.54
673,168.497,56.52
685,170.997,56.51
692,173.247,56.5
705,176.497,56.44
714,178.747,56.44
724,181.247,56.42
735,183.997,56.4
746,186.747,56.38
754,188.747,56.37
766,191.497,56.34
774,193.747,56.33
787,196.747,56.29
796,198.997,56.32
803,200.747,56.27
813,203.247,56.23
822,205.747,56.27
834,208.497,56.2
844,210.997,56.21
852,213.247,56.23
863,215.997,56.17
872,218.247,56.17
884,221.247,56.16
890,222.747,56.13
902,225.747,56.14
912,228.247,56.08
924,231.246,56.09
934,233.746,56.08
940,235.246,56.24
950,237.496,56.01
961,240.496,56
971,242.996,56
979,244.996,55.99
994,247.996,55.96
};
  \end{axis}
\end{tikzpicture} \hfill{}
\begin{tikzpicture}
  \begin{axis}[change y base, change x base, %
    y unit=MB, %
    width=0.45\textwidth, xlabel=compactness~$c$, ylabel=memory usage]

    \addplot[color=black,mark=+,only marks, each nth point=2] table
    [x=M, y=Mem, col sep=comma] {M,Mem,T
100,25.2496,56.11
110,27.7496,55.84
120,30.2495,55.67
130,32.7495,56.65
140,35.2495,56.58
150,37.7494,56.59
160,40.2494,56.47
170,42.7493,56.4
180,45.2493,56.3
190,47.4993,56.18
199,49.9992,56.11
210,52.7492,56.02
220,55.2492,55.93
230,57.7491,55.85
240,60.2491,55.78
249,62.499,55.67
259,64.999,56.64
270,67.749,56.64
280,69.9989,56.68
289,72.2489,56.61
300,75.2489,56.57
308,77.2488,56.55
316,79.2488,56.54
329,82.4987,56.62
339,84.9987,56.43
350,87.4987,56.39
360,89.9986,56.34
369,92.4986,56.32
378,94.7486,56.31
388,97.2485,56.25
399,99.7485,56.17
407,101.998,56.16
418,104.748,56.11
429,107.248,56.09
438,109.748,56.02
449,112.248,55.98
456,114.248,55.96
469,117.248,55.94
480,120.248,55.9
487,121.998,55.82
497,124.498,55.81
509,127.248,55.75
518,129.498,56.66
528,131.998,59.52
537,134.498,56.66
548,137.248,56.69
558,139.498,56.67
568,141.998,56.68
576,144.248,56.65
585,146.498,56.65
596,149.248,56.67
606,151.748,56.61
615,153.998,56.66
625,156.248,56.62
637,159.498,56.59
648,162.248,56.77
656,164.247,56.56
667,166.497,56.54
673,168.497,56.52
685,170.997,56.51
692,173.247,56.5
705,176.497,56.44
714,178.747,56.44
724,181.247,56.42
735,183.997,56.4
746,186.747,56.38
754,188.747,56.37
766,191.497,56.34
774,193.747,56.33
787,196.747,56.29
796,198.997,56.32
803,200.747,56.27
813,203.247,56.23
822,205.747,56.27
834,208.497,56.2
844,210.997,56.21
852,213.247,56.23
863,215.997,56.17
872,218.247,56.17
884,221.247,56.16
890,222.747,56.13
902,225.747,56.14
912,228.247,56.08
924,231.246,56.09
934,233.746,56.08
940,235.246,56.24
950,237.496,56.01
961,240.496,56
971,242.996,56
979,244.996,55.99
994,247.996,55.96
};
  \end{axis}
\end{tikzpicture}
\caption{Dependence of running time  and space requirements  of the dynamic program~(\ref{DP-Qp}, \autoref{sec:DP}) for \MColIS{} on the compactness~$c$ of the input interval graph, each having $10^5$~intervals colored using subsets of~$\{1,\dots,15\}$.}
\label{fig:vary-c}
\end{figure*}
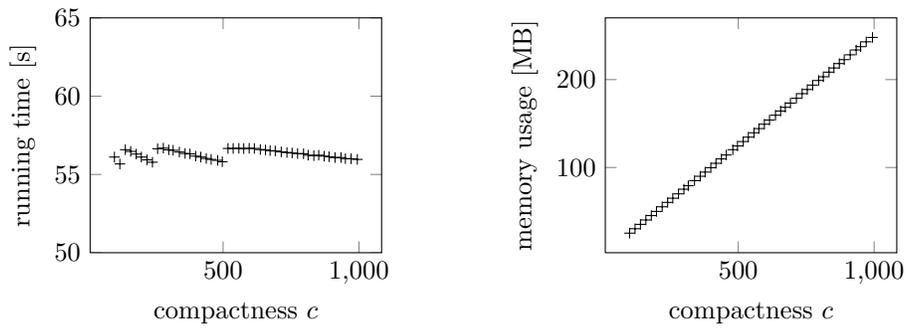

For the third data set, we finally let the compactness~$c$ vary between
$10^2$ and $10^3$. We again fixed $\gamma=15$ and
$n=10^5$. \autoref{fig:vary-c} shows the linear dependence of the memory
usage on the compactness~$c\approx\ell+1$. In contrast, the running time
remains roughly constant with increasing~$c$.  The observed local minima
of the running time are exactly at those values of~$c$ where $c$~is a
power of two.  In this case, we observed that the time spent per table
look-up decreases.  We suspect that this has technical reasons.

\paragraph{Summary.}\looseness=-1 The running time and memory usage of the algorithm on randomly generated data very reliably behave as predicted by \autoref{thm:mcolis-fpt} and most likely scale to larger data. We have seen that on moderate values of~$\gamma\leq15$, the algorithm can solve instances with up to $5.5\cdot 10^5$~intervals in a time frame of about five minutes. However, in application data, like for example from the steel manufacturing application of \citet{HKML11}, the number of colors can be much higher. To efficiently solve such instances with our algorithm, it is crucial that these instances have a low maximum number~$Q$ of ``\emph{live} colors'', that is, these instances must be more structured than our randomly generated interval graphs.

\section{Conclusion}
\label{sec:discussion}
We charted the complexity landscape of {\sc Independent Set} on
subclasses of 2-union graphs, which are of relevance for applications in
scheduling, and which generalize interval graphs. Our focus was on
determining the complexity of finding exact solutions, whereas, so far,
approximation algorithms have been much better researched in the literature~\citep{BNR96,Spi99,BHNSS06,COR06}.

Besides hardness results from
our complexity dichotomy, we provided first results on effective
polynomial-time preprocessing (kernelization) in this context.  We also
developed encouraging algorithmic results and evaluated them
experimentally, which might find use in practical
applications. %

For future work, it would be interesting to determine whether \MCIS{} is fixed-parameter tractable with respect to the ``$M$-compositeness'' parameter that is small in the steel manufacturing application considered by \citet{HKML11}.  Moreover, it seems worthwhile trying to speed up our randomized algorithm for \JIntSel{} (\autoref{thm:colorcoding}) using the algebraic techniques described by \citet{KW09}.

\paragraph{Acknowledgments.}
  We thank Michael Dom and Hannes Moser for discussions on coil coating, which initiated our investigations on \MCIS{}, as well as Wiebke Höhn for providing details regarding the application of \MCIS{} in steel manufacturing. %

  René van Bevern was supported by the Deutsche Forschungsgemeinschaft (DFG), project DAPA, NI~369/12. Part of the work was done while being supported by DFG project AREG, NI~369/9. Mathias Weller was supported by the DFG, project DARE, NI~369/11.

\bibliographystyle{abbrvnat}
\bibliography{mcis}

\end{document}